%% file: arxiv_main.tex
\documentclass[hidelinks]{article}
\usepackage[a4paper, total={5.5in, 8.5in}]{geometry}
\usepackage{appendix}
\usepackage{comment}

\usepackage{hyperref}
\usepackage{graphicx} 
\usepackage{amsthm}
\usepackage{amsmath}
\usepackage{amsfonts}

\usepackage[authoryear]{natbib}
\usepackage{bbm}
\usepackage{enumitem}

\usepackage{caption}
\usepackage{subcaption}
\newtheorem*{remark}{Remark}
\newtheorem{assumption}{Assumption}
 \usepackage{lineno} 
\newtheorem{theorem}{Theorem}
\newtheorem{proposition}{Proposition}
\newtheorem{corollary}{Corollary}
\newtheorem{lemma}{Lemma}
\newcommand{\R}{\mathbb{R}}
\newcommand{\E}{\mathbb{E}}

\usepackage{tcolorbox}
\usepackage{authblk}
\usepackage{wrapfig}
\newtcolorbox{mybox}{colback=blue!5!white,colframe=blue!75!black}
\providecommand{\keywords}[1]
{
  \small	
  \textbf{\textit{Keywords---}} #1
}
\title{Sharp Bounds on the Variance of General Regression Adjustment in Randomized Experiments}

\date{October 2026}
\author[1]{Jonas M. Mikhaeil\footnote{Email for correspondence: j.mikhaeil@columbia.edu}}
\author[2]{Donald P. Green}

\affil[1]{Department of Statistics, Columbia University, New York}
\affil[2]{Department of Political Science, Columbia University, New York}

\begin{document}

\maketitle

\input{paper/abstract}


\input{paper/introduction}
\input{paper/setup}

\input{paper/SharpVarianceBounds}

\input{paper/Estimation}
\input{paper/simulation}
\input{paper/empiricalExamples}
\input{paper/conclusion}
\section*{Data Availability and Code}
Code to reproduce our results is available under 
\href{https://github.com/JonasMikhaeil/SharpVarianceBounds}{https://github.com/ JonasMikhaeil/Sharp\\VarianceBounds}. The Fluid and Catheter Treatment Trial (FACTT) dataset is available from the Biologic Specimen and Data Repository Information Coordinating Center (BioLINCC) at \href{https://biolincc.nhlbi.nih.gov/studies/factt/}{https://biolincc.nhlbi.nih.gov/studies/factt/}. Access requires submission of a request and approval by BioLINCC.

\section*{Acknowledgments}
We thank Christopher Harshaw, Bodhisattva Sen, Nabarun Deb, Nicole Pashley, Getoar Sopa, and Andres Valdevenito for fruitful discussions.  We thank Melissa Michelson for providing replication data.
\bibliographystyle{apalike}
\bibliography{references}
\input{paper/appendix}

\end{document}

%% file: paper/abstract.tex
\begin{abstract}
A growing statistical literature focuses on causal inference in the context of experiments where the target of inference is the average treatment effect in a finite population and random assignment determines which subjects are allocated to one of the experimental conditions.  In this framework, variances of average treatment effect estimators remain unidentified because they depend on the covariance between treated and untreated potential outcomes, which are never jointly observed. Conventional variance estimators are upwardly biased. Recent work provides an estimator for the variance of the difference-in-means estimator that is asymptotically sharp.  The practical relevance of this result, however, remains limited, as researchers often rely on some form of covariate adjustment, such as linear regression, when estimating the average treatment effect.  Adapting propositions from empirical process theory, we provide asymptotically sharp variance estimators for general regression adjustment. We apply these results to linear regression adjustment and show benefits both in a simulation and in three empirical applications drawn from different disciplines. 
\end{abstract}

%% file: paper/introduction.tex
\section{Introduction}
Although a century has passed since the publication of \citet{Splawa-Neyman_Dabrowska_Speed_1990}, the finite population framework for analyzing randomized experiments remains a vibrant topic, perhaps because it cleanly differentiates between the analysis of an experiment on a given set of subjects and extrapolations to other populations.  Recent scholarship has made important contributions to the understanding of the finite population setting (cf. central limit theorems in \citep{Li_Ding_CLT,Schochet_2022_CLT,Liu_Ren_Yang_2022}, yet variance estimation remains a lingering concern. Neyman's original result noted the fact that the variance of the difference-in-means estimator is not point-identified except under special conditions, such as homogeneous treatment effects among all subjects. The widely-used classical variance estimator, described below, is unbiased under homogeneous treatment effects but positively biased under heterogeneous treatment effects. The classical variance estimator is recommended in textbook treatments \citep{ding2023course,Imbens_Rubin_2015,gerber2012field} on the grounds that it offers a ``conservative'' assessment of variance, which in turn implies that confidence intervals will have nominal coverage or greater.

Well aware of this problem of the unidentified variance of the average treatment effect estimator, \citet{Splawa-Neyman_Dabrowska_Speed_1990} derived an alternative variance estimator based on the Cauchy-Schwarz inequality that provides a less conservative upper bound on the variance. 
The Cauchy-Schwarz inequality, however, may be unnecessarily loose.
\citet{Aronow_Green_Lee_2014}, generalizing a result in \citep{Robins_1988}, derive a variance estimator based on sharp bounds on the estimator's variance.
Aronow et al. show that this estimator renders a less conservative variance estimate than the classic formula as well as the bounds proposed by \citet{Splawa-Neyman_Dabrowska_Speed_1990}.

The practical utility of these sharp bounds, however, remains limited. As proposed by \citet{Aronow_Green_Lee_2014}, they are only available for the difference-in-means estimator. 
This drastically limits their applicability, as randomized controlled trials
often use some form of covariate adjustment, such as weighting \citep{schochet2023designbased,Zubizarreta2023Handbook}, linear regression \citep{Lin_reg_adj} or machine learning methods \citep{Bloniarz_2016,Wu_Gagnon_2018,Su_Mou_Ding_Wainwright_2023} in order to improve precision. 
To fill this gap, the present paper derives a sharp variance bound and an accompanying variance estimator for general regression adjustment.
To illustrate our results, we explicitly apply it to the variance of linear regression adjustment \citep{Lin_reg_adj}, see Section \ref{Sec:LinRegAdj}. In Appendix \ref{Sec:DecorRegAdj}, we apply our results in the context of a decorrelation method for general regression adjustment \citep{Su_Mou_Ding_Wainwright_2023}.
We provide an R package
to make our sharp variance estimators accessible to practitioners.

The sharp variance estimators proposed here are guaranteed to be no larger than conventionally used variance estimators, and they potentially offer improvements under heterogeneous treatment effects.
As such, they are well-suited to randomized trials in which treatment effects are suspected to be heterogeneous ex ante, e.g., when pre-analysis plans propose to investigate heterogeneous effects across subgroups.  Given the many literature reviews in domains such as biomedical research \citep{Kent_Steyerberg_Klaveren_2018},
policy analysis \citep{Ferraro_Miranda_2013}, political science \citep{Kertzer_2022}, evaluation research \citep{Smith_2022}, psychology \citep{Hickin_Käll_Shafran_Sutcliffe_Manzotti_Langan_2021}, and behavioral science more generally \citep{Bryan_Tipton_Yeager_2021}
that offer theoretical and empirical grounds for expecting heterogeneous effects, it appears that many studies would benefit from sharp variance estimation.  Even if the gains in precision from the sharp bounds estimator are often modest -- we report benefits of up to 15 percent compared to conventional variance estimates -- the cumulative gains are sizable given the vast number of studies that could potentially benefit from more accurate variance estimation.

The operating characteristics of our sharp variance estimator for linear regression adjustment are demonstrated via simulations. We also show the potential practical benefits of this approach by reanalyzing three experiments from a range of different fields.  The first tests the ``anchoring effect'' that has long figured prominently in behavioral economics \citep{Lee_22}; the second assesses the effects of exposure to the TV sitcom \emph{Little Mosque on the Prairie} on attitudes toward Arabs \citep{Murrar_18}; the third revisits a biomedical trial that evaluates a treatment for acute lung injury \citep{BioExample}. 
We conclude by discussing the advantages and limitations of sharp bounds for variance estimation in the context of regression-adjusted estimation of average treatment effects.
\\
\textbf{Related work:}
The design-based, or finite population, perspective on causal inference that our paper takes is widely adopted in the literature \citep{Rosenbaum_2002,Imbens_Rubin_2015,Li_Ding_CLT,Lin_reg_adj}. \citet{Abadie_20} discuss the benefits of design-based inference in the context of regression analysis.

As noted above, the history of design-based inference for randomized experiments goes back to \citet{Splawa-Neyman_Dabrowska_Speed_1990}, who derives variance bounds based on the Cauchy-Schwarz inequality to allow for conservative inference despite the unidentified nature of the variance of treatment effect estimators. Generalizing \citep{Robins_1988}, \citet{Aronow_Green_Lee_2014} derive sharp variance bounds for the difference-in-means estimator for a randomized experiment with two arms.

Drawing on the sharp bounds result in \citet{Aronow_Green_Lee_2014}, \citet{Imbens_Menzel_2021} consider the least-favorable coupling of treated and untreated potential outcomes given the marginals of these potential outcomes to design a ``causal bootstrap'' procedure.

\citet{Yu_2025} extend sharp variance estimators to stratified experiments.
Moving beyond Fréchet-Hoeffding bounds, which offer a solution to the optimization problem defining sharp variance bounds for two-armed experiments, multi-marginal optimal transport can be used to define sharp variance bounds when there are more than two potential outcomes per unit \citep{gao_24}, as in the case of multi-arm trials or factorial designs.

Improving variance bounds for complicated designs and under interference is an open problem. In this context, \citet{Harshaw_26}  investigate optimized variance bounds for linear treatment effect estimators in the class of bounds that are quadratic forms in the potential outcome vectors.

%% file: paper/setup.tex
\section{Setup}
\label{Sec:Setup}
Consider a finite population $\mathcal{U}_N = \{(Y_i(1),Y_i(0),X_i)\}_{i=1}^N$. Invoking the Stable Unit Treatment Value Assumption, we denote with  $Y_i(1)$ and $Y_i(0)$ the potential outcomes \citep{Splawa-Neyman_Dabrowska_Speed_1990,Rubin_1974} under treatment and control. The vector $X_i = (x_{i1},\dots,x_{ik})^\top $ comprises pre-treatment covariates (for convenience, we will assume $\bar X = 0$ throughout). Taking a design-based perspective \citep{Abadie_20,Li_Ding_CLT,Lin_reg_adj}, we treat the finite population  $\mathcal{U}_N$ as fixed, considering only the randomness due to the treatment assignment $Z = (Z_1,\dots,Z_N)^\top$, which is defined such that $Z_i = 1$ if the unit is treated and $Z_i = 0$ otherwise. 
We are interested in determining the average treatment effect
\begin{align}
    \tau = \frac{1}{N}\sum_{i=1}^N \big (Y_i(1) - Y_i(0)\big)  = \bar{Y}(1) - \bar{Y}(0)
\end{align}
of the finite population.
We consider the case of a completely randomized experiment (CRE) in which $n_1$ of the $N$ units are randomly assigned into the treatment group and the remaining $n_0$ units receive the control. In the Appendix, we show that all our results also apply to Bernoulli randomized experiments (BRE) in which units are assigned into the treatment arm with probability $\pi_1$. The 
regression-adjusted estimator
\begin{align}
\label{Eq:GenRegAdjOracle}
    \hat \tau_N(f_1,f_0) &= \frac{1}{n_1} \sum_{i=1}^N Z_i(Y_i(1) - f_1(X_i)) - \frac{1}{n_0} \sum_{i=1}^N (1-Z_i)(Y_i(0) - f_0(X_i)) \\ \nonumber
    &+ \frac{1}{N} \sum_{i=1}^N (f_1(X_i) - f_0(X_i)),
\end{align}
with fixed outcome models $f_1: \R^k \rightarrow \R$ and $f_0: \R^k \rightarrow \R$
is unbiased for the average treatment effect $\tau$. For $f_0 = f_1\equiv  0$, this estimator reduces to the difference-in-means estimator. In Section \ref{Sec:LinRegAdj}, we apply our results explicitly to linear regression adjustment.
The goal of regression adjustment is to find outcome models such that the \emph{population-adjusted} potential outcomes
\begin{align}\nonumber
    \varepsilon_i(q) := Y_i(q) - f_q(X_i) \hspace{1cm} \text{for } q \in \{0,1\}
\end{align}
are as close to zero as possible. 

The variance of this fixed outcome model estimator is
\begin{align}\label{Eq:VarOracle}
      \text{Var}(\hat \tau_N(f_1,f_0)) = \frac{1}{N} \bigg(\frac{n_0}{n_1} S^2\big(\varepsilon(1)\big) + \frac{n_1}{n_0} S^2 \big(\varepsilon(0)\big)  
       +  2 S\big(\varepsilon(1),\varepsilon(0)\big)\bigg),
\end{align}
where $S^2(\varepsilon(q)) = \frac{1}{N-1}\sum_{i=1}^N (\varepsilon_{i}(q) -\bar \varepsilon(q) )^2$ for $q =0,1$ and $S(\varepsilon(1),\varepsilon(0)) = \frac{1}{N-1}\sum_{i=1}^N (\varepsilon_{i}(1) -\bar \varepsilon(1) ) (\varepsilon_{i}(0) -\bar \varepsilon(0))$ with $\bar \varepsilon(q) = \frac{1}{N}\sum_{i=1}^N \varepsilon_i(q)$. We refer to the last term $S(\varepsilon(1),\varepsilon(0))$ including both adjusted potential outcomes as the cross-term between adjusted potential outcomes. 

Although the average treatment effect is point-identified, the variance of the effect estimator is not. This lack of identification stems from the fact that the cross-term between the adjusted potential outcomes $S(\varepsilon(1),\varepsilon(0))$ remains unidentified because $\varepsilon_i(1)$ and $\varepsilon_i(0)$ are never observed jointly. 

Due to the unidentified nature of the variance of the effect estimator, inference is only possible by resorting to identified upper bounds of the variance, which we will simply refer to as variance bounds.
Following \citet{Splawa-Neyman_Dabrowska_Speed_1990}, a variance bound based on the Cauchy-Schwarz inequality and the arithmetic mean-geometric mean (AM-GM) inequality $2 S(\varepsilon(1),\varepsilon(0)) \leq 2 \sqrt{S^2(\varepsilon(1))S^2(\varepsilon(0))} \leq S^2(\varepsilon(1))+S^2(\varepsilon(0)) $ can be derived
\begin{align}
\label{Eq:Neyman}
   \text{VB}(\hat \tau_N(f_1,f_0))^{\text{Conv}} = \bigg(\frac{1}{n_1} S^2(\varepsilon(1))+\frac{1}{n_0} S^2(\varepsilon(0))\bigg).
\end{align}
In the case of difference-in-means and linear regression adjustment, this variance bound leads to the conventionally used variance estimators (see Equation \ref{eq:convVarLinReg}).

The outcome models $f_1$ and $f_0$ need to be determined based on data.
If all potential outcomes were observable, we could minimize the above variance bound by choosing outcome models given a function class $\mathcal{F}$ via least-squares
\begin{align}\nonumber
    f_q \in \text{arg}\min_{f_q \in \mathcal{F}}\bigg \{\sum_{i=1}^N (Y_i(q) - f_q(X_i) )^2 \bigg\}
\end{align}
on the control and treated potential outcomes separately.

In practice, we face the difficulty that only the potential outcomes associated with the allocated treatment are observed. Due to this inherent problem of missingness,
we need to work with an estimator of the form
\begin{align}\nonumber
\hat \tau_N 
\, & = \, \frac{1}{n_1} \sum_{i=1}^N Z_i \hat \varepsilon_{i}(1) - \frac{1}{n_0} \sum_{i=1}^N (1-Z_i) \hat{\varepsilon}_{i}(0) + \frac{1}{N} \sum_{i=1}^N(\hat f_{1}(X_i) - \hat f_{0}(X_i)) , 
\end{align}
where $\hat f_1$ is an estimate of $f_1$ based on the treated units and $\hat f_0$ is an estimate of $f_0$ based on the control units. We refer to $\hat{\varepsilon}_{i}(q)= Y_i(q) - \hat f_{q}(X_i)$ as the \emph{sample-adjusted} potential outcomes. 

In the cases considered in this paper, the additional variation introduced by the estimates $\hat f_1$ and $\hat f_0$ is asymptotically negligible, which allows us to work with plug-in estimators for the variance bound of the fixed outcome model estimator.

Following \citet{Splawa-Neyman_Dabrowska_Speed_1990}, one can use the Cauchy-Schwarz inequality to arrive at the following variance estimator, 
\begin{align}
\label{Eq:CSVarEst}
      \widehat{\text{Var}}(\hat \tau_N)^{\text{CS}} = \frac{1}{N} \bigg(\frac{n_0}{n_1} S^2\big(\hat \varepsilon(1)\big) + \frac{n_1}{n_0}  S^2\big(\hat \varepsilon(0)\big) +  2\sqrt{S^2\big(\hat \varepsilon(1)\big)S^2\big(\hat \varepsilon(0)\big)} \bigg),
\end{align}
which we call the Cauchy-Schwarz variance estimator. Here $S^2(\hat \varepsilon(1)) =\frac{1}{n_1-1}\sum_{i=1}^N Z_i (\hat \varepsilon_{i}(1) - \hat{\bar{\varepsilon}}(1))^2$ and $S^2(\hat \varepsilon(0)) =\frac{1}{n_0-1}\sum_{i=1}^N (1-Z_i) (\hat \varepsilon_{i}(0) - \hat{\bar{\varepsilon}}(0))^2$ with $\hat{\bar{\varepsilon}}(1) = \frac{1}{n_1}\sum_{i=1}^N Z_i \hat \varepsilon_i(1)$ and $\hat{\bar{\varepsilon}}(0) = \frac{1}{n_0}\sum_{i=1}^N (1-Z_i) \hat \varepsilon_i(0)$.
In practice, the more convenient (but often more conservative) estimator based on the further application of the AM-GM inequality is used
\begin{align}
\label{eq:convVarEst}
      \widehat{\text{Var}}(\hat \tau_N)^{\text{Conv}} =  \frac{1}{n_1} S^2(\hat \varepsilon(1)) + \frac{1}{n_0} S^2(\hat \varepsilon(0)),
\end{align}
which we call the conventional variance estimator.
While linear regression adjustment will serve as our main illustration, our results are more general. Our results on sharp variance bound estimators (Proposition \ref{Prop:ConsistencyOfSharpBounds}) apply to general regression adjustment as discussed above. 
In Appendix \ref{Sec:DecorRegAdj}, we give an additional example and apply our Proposition to a decorrelation estimator for general regression adjustment developed in \citep{Su_Mou_Ding_Wainwright_2023}. While the asymptotics of general regression adjustment have received much attention recently \citep{Bloniarz_2016,Wu_Gagnon_2018,Guo_Basse_2023}, we choose their method as an example because their asymptotic theory only requires $o_p(1)$ consistency for the outcome estimates, matching the mild conditions of our theory.

We stress, however, that linear regression adjustment and the decorrelation method for general regression adjustment only serve as examples of our sharp bounds and their estimators. Our results on variance estimation apply broadly to general regression adjustment.

%% file: paper/SharpVarianceBounds.tex
\section{Sharp Variance Bounds}
In the previous section, we have seen that the variance of effect estimators in the design-based setting remains unidentified. This necessitates the use of identified variance bounds. Conventional variance bounds rely on the application of the Cauchy-Schwarz and AM-GM inequalities. This naturally raises the question of whether tighter -- or even sharp -- variance bounds exist. This section answers this question in the affirmative.

\subsection{Population CDFs and Sharp Variance Bounds}
While we remain in the design-based setting and treat the (population-adjusted) potential outcomes as fixed, it is helpful to introduce the \emph{notation} of population CDFs to utilize concepts from probability theory. For example, the use of these CDFs opens the door to results from empirical process theory.
We provide a P-Glivenko-Cantelli result for empirical distributions in the finite population setting in Appendix \ref{App:EmpProc}. While these results may be useful more broadly in the finite population setting, here they will be crucial to prove consistency of sharp variance bound estimators for general regression adjustment, see Section \ref{Sec:Estimation}.

We begin by defining the population joint CDF of the population-adjusted potential outcomes 
\[
\Gamma_N(\xi_1,\xi_0) = \frac{1}{N} \sum_{i=1}^N \mathbbm{1}\{\varepsilon_i(1) \leq \xi_1 , \varepsilon_i(0) \leq \xi_0\}.
\]
This joint CDF is unidentified because $\varepsilon_i(1)$ and $\varepsilon_i(0)$ are never observed jointly. The marginal CDFs, however, 
\begin{align}
\label{eq:marginalCDFS}
    G_N(\xi)  = \frac{1}{N}\sum_{i=1}^N \mathbbm{1}\{\varepsilon_i(1)\leq \xi\} \enspace \text{and} \enspace 
    F_N(\xi) = \frac{1}{N}\sum_{i=1}^N \mathbbm{1}\{\varepsilon_i(0)\leq \xi\}
\end{align}
are identified. These finite population distributions are useful as their moments match the finite population quantities that appear in the variance expression, see Equation \ref{Eq:VarOracle}.
Throughout, we will use the following notation $\mathbb{E}_{G}[\varepsilon] \equiv \mathbb{E}_{\varepsilon \sim G}[\varepsilon]$. Then, for example, $\mathbb{E}_{G_N}[\varepsilon(1)] = \frac{1}{N}\sum_{i=1}^N \varepsilon_i(1) = \bar{\varepsilon}(1)$ and 
$
\frac{N}{N-1}\text{Var}_{G_N}(\varepsilon(1)) = \frac{1}{N-1}\sum_{i=1}^N \big(\varepsilon_i(1) - \bar{\varepsilon}(1)\big)^2 = S^2(\varepsilon(1))$.

This notation allows us to re-express the variance of the fixed outcome model estimator
\begin{align*}
  \text{Var}(\hat \tau_N(f_1,f_0)) & = \frac{1}{N-1} \bigg( \frac{n_0}{n_1} \text{Var}_{G_N}\big(\varepsilon(1)\big) + \frac{n_1}{n_0} \text{Var}_{F_N}\big(\varepsilon(0)\big) 
       +  2 \text{Cov}_{\Gamma_N}\big(\varepsilon(1),\varepsilon(0)\big)\bigg) \\ 
       &=:  \text{Var}_{\Gamma_N}(\hat \tau_N(f_1,f_0)) ,
\end{align*}
where $\text{Cov}_{\Gamma_N}(\varepsilon(1),\varepsilon(0)) \equiv \E_{(\varepsilon(1),\varepsilon(0))\sim \Gamma_N}[\varepsilon(1)\varepsilon(0)]-\E_{\varepsilon(1)\sim G_N}[\varepsilon(1)]\E_{\varepsilon(0)\sim F_N}[\varepsilon(0)]$.

Note that the variance is completely determined by the unidentified joint CDF $\Gamma_N$. Any joint CDF $\gamma$ on the population-adjusted potential outcomes thus immediately implies a variance of the effect estimator $\text{Var}_\gamma(\hat \tau_N(f_1,f_0))$. 

Holding the identified marginal CDFs $G_N$ and $F_N$ fixed, we ask what the largest possible variance of the effect estimator is over all unidentified joint CDFs compatible with the identified marginals. That is, 
\begin{align}
\label{eq:sharp}
\text{SVB}_N : = \max_{\gamma \in \Pi(G_N,F_N)} \text{Var}_\gamma(\hat \tau_N(f_1,f_0)),
\end{align}
where $\Pi(G_N,F_N)$ is the set of all joint CDFs with marginal CDFs $G_N$ and $F_N$. We call such variance bounds \emph{sharp}. 
Similarly, a sharp lower variance bound can be defined, see Appendix \ref{App:LowerBounds}.
\subsection{Fréchet-Hoeffding copula bounds}
\label{Sec:FHbounds}
We follow the literature \citep{Robins_1988,Aronow_Green_Lee_2014} in using Fréchet-Hoeffding copula bounds to determine the sharp variance bound in Equation \ref{eq:sharp}.

 The works of \citet{Frechet_1960} and \citet{Hoeffding_1940} provide bounds on copulas:
\begin{theorem}[Fréchet–Hoeffding bounds, see \citep{Jaworski_2010}]
\label{ref:FH}
   For any $d$-dimensional copula $C$ and any $u = (u_1,\dots,u_d) \in [0,1]^d$, the following bounds hold
    \begin{align}\nonumber
    W(u) \leq C(u) \leq M(u), 
    \end{align}
    where $W(u) := \max(1-d +\sum_{i=1}^d u_i,0)$ and $M(u) := \min(u_1,\dots,u_d)$ are the Fréchet–Hoeffding lower and upper bounds, respectively. 
\end{theorem}
The upper bound is always a copula and hence sharp. 
In two dimensions, the upper copula bound leads to the following Lemma \citep{Puccetti_Wang_2015}:
\begin{lemma}
\label{lemma:Hoeffding}
For two random variables  $X_1$ and $X_2$ with marginal CDFs $G$ and $F$, we have
    \[
    \max_{\gamma \in \Pi(G,F)} \E_{(X_1,X_2) \sim \gamma}[X_1X_2] = \E_{(X_1,X_2) \sim \Gamma^U}[X_1X_2]
    \]
with $\Gamma^U(x_1,x_2) = \min\{G(x_1),F(x_2)\}$.
   The upper bound is attained when $X_1$ and $X_2$ are comonotonic.
\end{lemma}

The application to randomized experiments is immediate. 
We define the joint CDF of extremal dependence $\Gamma^U_N(\xi_1,\xi_0) = \min\{G_N(\xi_1),F_N(\xi_0)\}$ based on the marginal CDFs of the population-adjusted potential outcomes in Equation \ref{eq:marginalCDFS}.
 With Lemma \ref{lemma:Hoeffding}, we have that the sharp variance bound is given by
\begin{align}
\label{Eq:SVB_CRE}
   \text{SVB}_N & = \frac{1}{N-1} \bigg( \frac{n_0}{n_1} \text{Var}_{G_N}\big(\varepsilon(1)\big) + \frac{n_1}{n_0} \text{Var}_{F_N}\big(\varepsilon(0)\big) 
       +  2 \text{Cov}_{\Gamma^U_N}\big(\varepsilon(1),\varepsilon(0)\big)\bigg).
\end{align}

%% file: paper/Estimation.tex
\section{Estimation of Sharp Variance Bounds}
\label{Sec:Estimation}
In the context of a randomized experiment, the marginal CDFs $G_N$ and $F_N$, while identified, remain unobserved. Instead, we estimate the marginals with the empirical distributions $\hat G_N(\xi) = 1/n_1 \sum_{i=1}^N Z_i\mathbbm{1}\{\hat \varepsilon_i(1)\leq \xi\}$ and $ \hat F_N(\xi) = 1/n_0 \sum_{i=1}^N (1-Z_i) \mathbbm{1}\{\hat \varepsilon_i(0)\leq \xi \}$. 

These yield the plug-in estimators for the sharp variance bound
 \begin{align}
 \label{Eq:estmVarBoundsReg}
 \widehat{\text{SVB}}_N = 
     \frac{1}{N-1} \bigg( \frac{n_0}{n_1} \text{Var}_{\hat G_N}\big(\hat \varepsilon(1)\big) + \frac{n_1}{n_0} \text{Var}_{\hat F_N}\big(\hat \varepsilon(0)\big) 
       +  2 \text{Cov}_{\hat \Gamma^U_N}\big(\hat \varepsilon(1),\hat \varepsilon(0)\big)\bigg),
 \end{align}
 \sloppy
 where we use the plug-in estimator of the joint of extremal dependence $\hat \Gamma^U_N(\xi_1,\xi_0) = \min\{\hat G_N(\xi_1),\hat F_N(\xi_0)\}$.
To study the asymptotic regime of these estimators of the sharp bounds, we make the following regularity condition:
 \begin{assumption}[Convergence of the joint distribution of the population-adjusted potential outcomes]
 \label{as:convergJoint}
 The joint distribution of the population-adjusted potential outcomes  $H_N$ converges weakly to a distribution $H$ with marginals $H(\cdot,\infty) = G(\cdot)$ and $H(\infty,\cdot)=F(\cdot)$.
 \end{assumption}

 Compared to \citet{Aronow_Green_Lee_2014}, we face the additional difficulty that the arguments of the indicator functions in $\hat G_N$ and $\hat F_N$ now depend on the treatment assignment $Z$ as well. This difficulty arises because we do not observe the population-adjusted potential outcomes $\varepsilon_{i}(q)$ based on oracle models $f_{1}$ and $f_{0}$ but only the sample-adjusted potential outcomes $\hat \varepsilon_{i}(q)$ that depend on the (random) estimates of the outcome models $\hat f_{1}$ and $\hat f_{0}$ fitted to the observed data. 
The following proposition shows that this difficulty can be overcome without additional assumptions (beyond classical regularity conditions). We show that the plug-in estimator for the sharp variance bound (Equation \ref{Eq:estmVarBoundsReg})
 is consistent.
 \begin{proposition}
 \label{Prop:ConsistencyOfSharpBounds} Let Assumption \ref{as:convergJoint} hold. Further, assume:
     \begin{enumerate}[label=(\alph*)]
     \item The proportion of treated units converges $\frac{n_1}{N} \rightarrow \rho \in (0,1)$.
         \item The estimates of the outcome models are $o_p(1)$ consistent, that is,
    \begin{align}\nonumber
       \big( \frac{1}{N} \sum_{i=1}^N (f_{q,N}(X_i) - \hat f_{q,N}(X_i))^2 \big)^{1/2}= o_p(1) \hspace{0.3cm} \text{for } q \in \{0,1\}.
    \end{align} 
         \item The population-adjusted potential outcomes $\varepsilon_{i,N}(q) = Y_i(q) - f_{q,N}(X_i)$ are uniformly square-integrable, that is
         \begin{align}\nonumber
              \sup_N \frac{1}{N} \sum_{i=1}^N \varepsilon_{i,N}(q)^2 \mathbbm{1}[\varepsilon_{i,N}(q)^2 \geq \beta] \rightarrow 0
         \end{align}
        as $\beta \rightarrow \infty$ for $q \in \{0,1\}$.
     \end{enumerate}
    Then
     \begin{align}\nonumber
         N \cdot \emph{\text{SVB}}_N \rightarrow \frac{1-\rho}{\rho} \emph{\text{Var}}_G[\varepsilon(1)] +  \frac{\rho}{1-\rho}\emph{\text{Var}}_{F}[\varepsilon(0)] + 2 \sup_{\gamma \in \Pi(G,F)}\emph{\text{Cov}}_\gamma(\varepsilon(1),\varepsilon(0))
     \end{align}
     and
     \[
     \widehat {\emph{\text{SVB}}}_N - \emph{SVB}_N =o_p(1/N).
     \]
 \end{proposition}
 \begin{proof}
    The proof of Proposition \ref{Prop:ConsistencyOfSharpBounds} can be found in Appendix \ref{Appendix:proof_prop1}. Notice that the proof strategy provided by \citet{Aronow_Green_Lee_2014} does not generalize to covariate adjustment because it requires uniform convergence of the empirical distributions. Our proof proceeds by using weaker metrics -- the Lévy distance and the bounded-Lipschitz distance -- that metrize weak convergence and thus allow us to establish consistency without requiring convergence of CDFs on discontinuity points. This turns out to be crucial when extending to covariate adjustment because of the additional challenge of the indicator functions in the empirical distributions being random. Our proof builds on a P-Glivenko-Cantelli result for the empirical distributions $\hat G_N$ and $\hat F_N$  that we derive in Appendix \ref{App:EmpProc}.
\end{proof}
\begin{remark}
    By using weaker metrics than total variation that metrize weak convergence, the proof strategy developed here does not require the marginals $G_N$ and $F_N$ to converge pointwise to $G$ and $F$ at their respective discontinuity points. Following the proof developed here, this assumption in \citep{Aronow_Green_Lee_2014} Proposition 1 can thus be relaxed. Our work also extends their results to BREs. 
\end{remark}

\subsection{Calculation of the Estimators for the Sharp Variance Bounds}
\label{Sec:Calc}
In this section, we will show how to compute the cross-term between adjusted potential outcomes under the joint distribution of extremal dependence. These, in turn, allow for the computation of our sharp variance bound estimator. 
The sharp bound on the second moment in Lemma \ref{lemma:Hoeffding} corresponds to random variables that are comonotonic. Defining the left-continuous inverse of CDFs as $F^{-1}(x) = \inf\{y : F(y) \geq x\}$, we can express comonotonicity of the sample-adjusted potential outcomes as $(\hat\varepsilon_{N}(1),\hat\varepsilon_{N}(0)) \sim (\hat G_N^{-1}(U),\hat F_N^{-1}(U))$, where $U \sim \text{Unif}[0,1]$ is uniformly random \citep{Puccetti_Wang_2015}. Hence
\begin{align}\nonumber
    \E_{\hat \Gamma^U_N}[\hat\varepsilon(1)\hat\varepsilon(0)] = \int_{0}^1 \hat G_N^{-1}(u)\hat F_N^{-1}(u)du.
\end{align}
Let $\hat\varepsilon_{(1)1}^{obs},\dots,\hat \varepsilon_{(n_1)1}^{obs}$ and $\hat \varepsilon_{(1)0}^{obs},\dots,\hat \varepsilon_{(n_0)0}^{obs}$ be the ordered sample-adjusted outcomes in the treatment and control group, respectively, and let $\{p_0,\dots p_P\}$ be the ordered and distinct elements of $ \{i/n_1\}_{i=0}^{n_1} \bigcup \{j/n_0\}_{j=0}^{n_0}$. Then 
\begin{align}
\label{Eq.CalcBounds}
      \E_{\hat \Gamma^U_N}[\hat\varepsilon(1)\hat\varepsilon(0)]  = \sum_{i=1}^P (p_i - p_{i-1})\hat \varepsilon_{(\lceil \tilde n_{1} p_i \rceil)1}^{obs}\hat \varepsilon_{(\lceil \tilde n_{0} p_i \rceil)0}^{obs}  .
\end{align}
 We provide code to calculate these sharp bounds in the form of an R package.

\subsection{Application to Linear Regression Adjustment}
 \label{Sec:LinRegAdj}
One common way of adjusting for covariates is choosing $\mathcal{F}$ to be the class of all linear outcome models (i.e., $f_q(X_i) = \alpha_q + \beta_q^\top X_i$). While this adjustment comes with finite-sample bias \citep{Freedman_2008}, \citet{Lin_reg_adj} showed that with large samples, linear covariate adjustment (with full treatment-covariate interactions) is at least as efficient as a difference-in-means estimator, which makes no use of the covariates. 
For linear outcome models, the treatment effect estimator in Equation \ref{Eq:GenRegAdjOracle} is equivalent to
\begin{align}
\label{Eq:RegAdj}
    \hat \tau_N(\beta_1,\beta_0) = \frac{1}{n_1} \sum_{i=1}^N Z_i \big(Y_i(1) - \beta_1^\top X_i \big) - 
    \frac{1}{n_0} \sum_{i=1}^N (1-Z_i)\big(Y_i(0) - \beta_0^\top X_i\big)
\end{align}
and is often referred to as Lin's estimator.
\citet{Li_Ding_CLT} show that 
the variance of $\hat \tau(\beta_1,\beta_0)$ is minimized by choosing $\beta_1$ and $\beta_0$ as the OLS coefficients of $X_i$ in the linear projection of the potential outcomes $Y(1)$ and $Y(0)$ onto $(1,X_i)$:
\begin{align}\nonumber
(\beta_{q},\alpha_{q}) = \underset{\beta_q,\alpha_q}{\text{argmin}}\sum_{i=1}^N\bigg(Y_i(q) - \alpha_q - \beta_q^\top X_i \bigg)^2 \, \, \, \text{for} \, \, q = 0,1.
\end{align}
These optimal coefficients remain unknown, and we have to make do with the estimated least-squares coefficients
\begin{align}\nonumber
(\hat \beta_{1},\hat \alpha_{1}) \, &=\,  \underset{\beta_1,\alpha_1}{\text{argmin}}\sum_{i=1}^N Z_i \bigg(Y_i(1) - \alpha_1 - \beta_1^\top X_i \bigg)^2 \, \, \, \text{and} \\\nonumber
(\hat \beta_{0},\hat \alpha_{0}) \, &=\,  \underset{\beta_0,\alpha_0}{\text{argmin}}\sum_{i=1}^N (1-Z_i) \bigg(Y_i(0) - \alpha_0 - \beta_0^\top X_i \bigg)^2.
\end{align}
 While using the estimates $\hat \beta_{1}$ and $\hat \beta_{0}$ introduces finite-sample bias of order $o_p(1/\sqrt{N})$ \citep{Lin_reg_adj}, \citet{Li_Ding_CLT} show that $\hat \tau_N(\beta_{1},\beta_{0})$ and $\hat \tau_N(\hat \beta_{1},\hat \beta_{0})$ have the same asymptotic distribution. Here the population-adjusted potential outcomes are $\varepsilon_{i}(q) = Y_i(q) - (\alpha_{q}+\beta_{q}^\top X_i)$, and the variance of $\hat \tau(\beta_{1},\beta_{0})$ can be expressed as
 \begin{align}
 \label{Eq:VarLinAdj}
     \text{Var}(\hat \tau(\beta_{1},\beta_{0})) 
     \, &= \, \frac{1}{N-1} \bigg(\frac{n_0}{n_1}\frac{1}{N}\sum_{i=1}^N \varepsilon_{i}(1)^2 + \frac{n_1}{n_0} \frac{1}{N}\sum_{i=1}^N\varepsilon_{i}(0)^2 + 2 \frac{1} {N}\sum_{i=1}^N\varepsilon_{i}(1)\varepsilon_{i}(0) \bigg), 
 \end{align}
 where we have relied on the fact that for linear regression the population-adjusted potential outcomes are mean zero.
 Typically, the (asymptotically) conservative estimator 
 \begin{align}
 \label{eq:convVarLinReg}
     \widehat{\text{Var}}(\hat \tau(\hat \beta_{1},\hat \beta_{0})) =  \frac{1}{n_1(n_1 -1)}\sum_{i=1}^N Z_i \hat \varepsilon_{i}(1)^2 +  \frac{1}{n_0(n_0-1)}\sum_{i=1}^N (1-Z_i)\hat \varepsilon_{i}(0)^2,
 \end{align}
 where $\hat \varepsilon_{i}(q) = Y_i(q) - (\hat \alpha_{q} + \hat \beta_{q}^\top X_i)$ are the sample-adjusted outcomes, is used. This variance estimator is equivalent to the conventional estimator (Equation \ref{eq:convVarEst}) in Section \ref{Sec:Setup}.
 \citet{Lin_reg_adj} showed that this conservative variance estimator can be conveniently approximated by Huber-White (EHW) robust standard errors. 

 Additionally, it can be shown that post-stratification \citep{Miratrix_Sekhon_Yu_2013} based on discrete covariates $C_i$ with $K$ categories is equivalent to linear regression adjustment with covariates $X_i = \big (\mathbbm{1}\{C_i = 1\} - \pi_1,\dots,\mathbbm{1}\{C_i = K-1\} - \pi_{K-1}\big)$, where $\pi_k$ is the proportion of units within stratum $k$ \citep{ding2023course}.

 In practice, Lin's estimator can be calculated by linearly regressing $Y^{obs}$ on $(1,Z,X,Z\times X)$ after mean-centering $X$. Lin's estimator $\hat \tau(\hat \beta_{1},\hat \beta_{0})$ of the average treatment effect is then given by the OLS coefficient of $Z$. 

The following Corollary applies Proposition \ref{Prop:ConsistencyOfSharpBounds} to the case of linear regression adjustment.
 \begin{corollary}
 \label{cor:linReg}
 Let Assumption \ref{as:convergJoint} hold. Further, assume the following classical regularity conditions:
 \begin{enumerate}[label=(\alph*)]
 \item The proportion of treated units converges $\frac{n_1}{N} \rightarrow \rho \in (0,1)$.
 \item Finite population variances and covariances among potential outcomes and covariates have limiting values.
    \item Finite population covariance of covariates and its limits are nonsingular.
    \item The potential outcomes and covariates have bounded fourth moments
\begin{align}\nonumber
      \frac{1}{N} \sum_{i=1}^N Y_i(0)^4 \leq L, \, \,  \frac{1}{N} \sum_{i=1}^N Y_i(1)^4 \leq L\text{, and }   \frac{1}{N} \sum_{i=1}^N x_{ki}^4 \leq L.
    \end{align}
\end{enumerate}
    Then the following holds
\begin{align}\nonumber
    \frac{(\hat \tau_N(\hat \beta_1,\hat \beta_0) - \tau)}{\sqrt{\kappa \widehat{\emph{\text{SVB}}}_N}} \overset{d}{\rightarrow} \mathcal{N}(0,1)
\end{align}
for the asymptotic distribution of Lin's estimator, where $\kappa \leq 1$ is the asymptotic ratio of the true variance to the sharp variance bound.
  \end{corollary}
  \begin{proof}
  Under these assumptions, \citet{Lin_reg_adj} establishes a CLT for Lin's estimator.
      The regularity conditions imply both Assumptions (b) and (c) of Proposition \ref{Prop:ConsistencyOfSharpBounds}. The corollary thus follows immediately. 
  \end{proof}
  \begin{remark}
    Using the asymptotic results above, researchers can use the 
    asymptotically conservative Wald-type confidence interval
       $\hat \tau_N(\hat \beta_1,\hat \beta_0) \pm z_{1-\alpha/2} \sqrt{\widehat{ \emph{\text{SVB}}}_N}$ using our sharp variance bound estimator.
  \end{remark}

When using linear regression to calculate Lin's estimator, the sharp variance bounds (Equation \ref{Eq:estmVarBoundsReg}) can be conveniently calculated based on the regression residuals $r$ from the regression of $Y^{obs}$ on $(1,Z,X,Z \times X)$. The residuals $r$ can be directly plugged into Equation \ref{Eq.CalcBounds}. 

%% file: paper/simulation.tex
\section{Simulation Results}
\label{Sec:SimRes}
While our sharp variance estimator is never larger than the conventional variance estimator, this section offers a simulation study showcasing that our estimator is less conservative under heterogeneous treatment effects.
Additionally, simulations allow us to compare the finite-sample properties of the available variance estimators.

We know that the conventional variance estimator is consistent in the case of  constant treatment effects. By extension, when the sample-adjusted potential outcomes are linearly dependent, the Cauchy-Schwarz inequality is exact.
In both cases, sharp variance bounds cannot provide any benefits. 
We thus set up a simulation that interpolates between these two extremes by introducing heterogeneity into the treatment effects. As such, the simulation serves as a simple illustration of a case in which sharp variance bounds outperform other variance estimators.

Let $p_i \overset{i.i.d}{\sim} \text{Bernoulli}(\theta)$, $e_i \overset{i.i.d} \sim \mathcal{N}(0,1) $, and $x_i \overset{i.i.d} \sim  \text{Unif}([-1,1])$. We then define the potential outcomes $Y_i(0) = x_i + e_i$ and
\begin{align}\nonumber
    Y_i(1) = \begin{cases}
        Y_i(0)  \hspace{1.5cm} &\text{if } p_i = 0 \\
        10+0.5e_i \, \, \, &\text{else}.
    \end{cases}
\end{align}
The parameter $\theta$ can be thought of as a heterogeneity parameter. In keeping with the finite-population framework, we only draw the potential outcomes $Y(0)$ and $Y(1)$ once. The randomness in the subsequent simulations stems purely from randomly drawing the treatment indicator $Z$, i.e., $Y_i^{obs} = Z_i Y_i(1) + (1-Z_i)Y_i(0)$ inherits its randomness only from $Z_i$. The parameter $\theta$ moves the simulation between the extremes of the sharp null ($\theta = 0$) and linear dependence between the adjusted potential outcomes ($\theta = 1$). 
For intermediate values of $\theta$, this data-generating process yields heterogeneous effects with a non-linear relationship between the adjusted potential outcomes $\varepsilon(1)$ and $\varepsilon(0)$.
\begin{figure}
     \centering
     \begin{subfigure}[b]{0.49\textwidth}
         \centering
         \includegraphics[width=\textwidth]{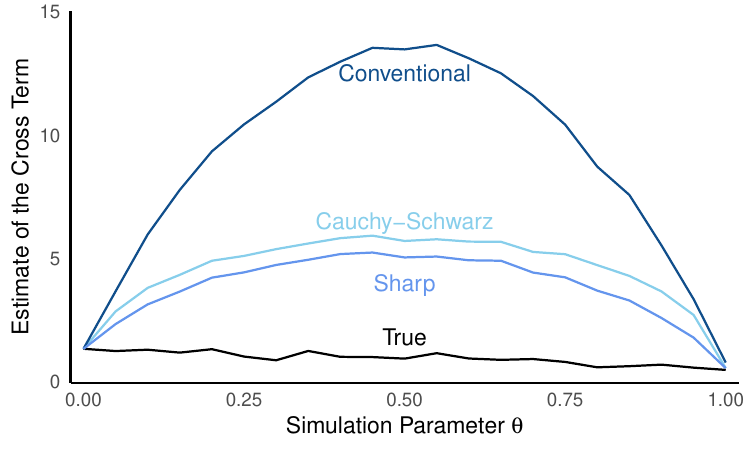}

     \end{subfigure}
     \hfill
     \begin{subfigure}[b]{0.49\textwidth}
         \centering
    \includegraphics[width=\textwidth]{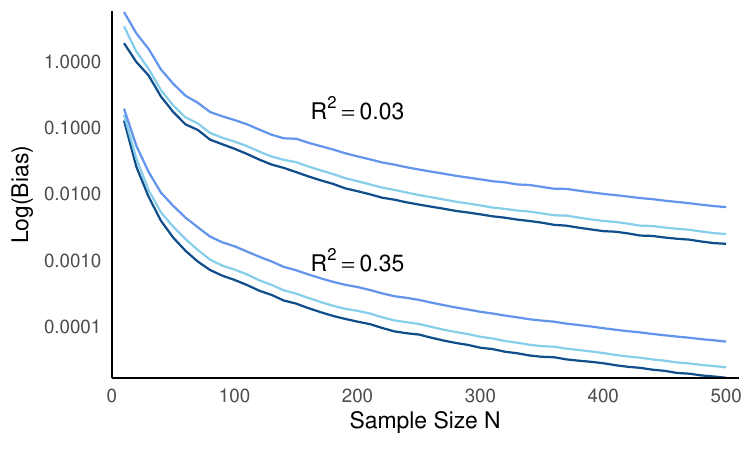}

     \end{subfigure}
     \hfill
        \caption{Simulation Results. (Left) Impact of the heterogeneity parameter $\theta$ (horizontal axis) on the conventional, Cauchy-Schwarz, and sharp bound estimates of the cross-term between the population-adjusted potential outcomes (vertical axis).  (Right) Comparison of the small-sample bias of the different variance estimators under the sharp null. We plot the log bias of the conventional, Cauchy-Schwarz, and sharp bound variance estimators for different levels of informativeness (in terms of R$^2$) of the covariates. The bias for all variance estimators shows the same scaling with sample size. }
        \label{fig:simulationResults}
\end{figure}

Figure \ref{fig:simulationResults} (Left) shows estimates of the unidentified cross-term between the adjusted potential outcomes in Equation \ref{Eq:VarLinAdj}. As expected, there are no benefits when using sharp variance bounds for either of the limiting cases (the sharp null $\theta = 0$ and linear dependence $\theta =1$).  For intermediate values of $\theta$, however, there are non-linear heterogeneous effects and clear benefits of the sharp variance estimator. Q-Q plots of the sample-adjusted outcomes in the treatment and control groups offer a practical tool to assess potential benefits of sharp variance bounds beyond the Cauchy-Schwarz variance estimator. In Figure \ref{fig:QQplots} in Appendix \ref{App:BenSharpBounds}, we show the Q-Q plot for simulations with $\theta=0.5$. We also discuss the conditions under which sharp bounds offer gains in precision. 

Figure \ref{fig:simulationResults} (Right) shows the small-sample bias (on a logarithmic scale) of the conventional, the Cauchy-Schwarz, and our sharp bound estimator of the variance as a function of sample size under the sharp null for different levels of informativeness of the covariates. 
While all three variance estimators show small-sample bias, this simulation shows that the bias of our sharp variance estimator is of similar order as the bias of the conventional (EHW robust) variance estimator and diminishes quickly. 

%% file: paper/empiricalExamples.tex
\section{Empirical Examples}
\label{Sec:Appl}
In this section, we apply our approach to experiments from an array of different fields. In all three examples, prognostic covariates improve the precision with which the ATE is estimated.  Using our sharp variance bounds for the regression-adjusted treatment effect estimator reduces the estimated variance in all three examples, highlighting the potential benefits of applying sharp bounds to covariate-adjusted estimators.
\subsection{Psychology Example: The Effect of Entertainment-Education on Prejudice}
\begin{figure}
     \centering
     \begin{subfigure}[b]{0.49\textwidth}
         \centering
    \includegraphics[width=\textwidth]{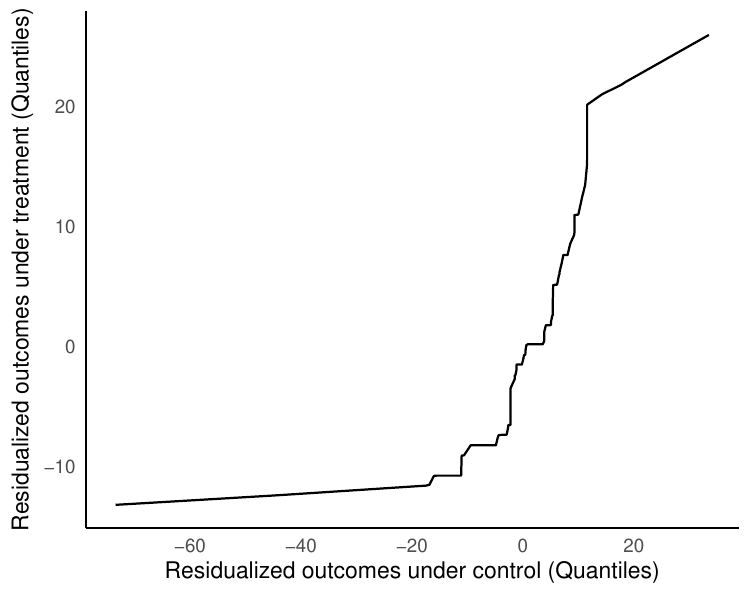}

     \end{subfigure}
     \hfill
     \begin{subfigure}[b]{0.49\textwidth}
         \centering
         \includegraphics[width=\textwidth]{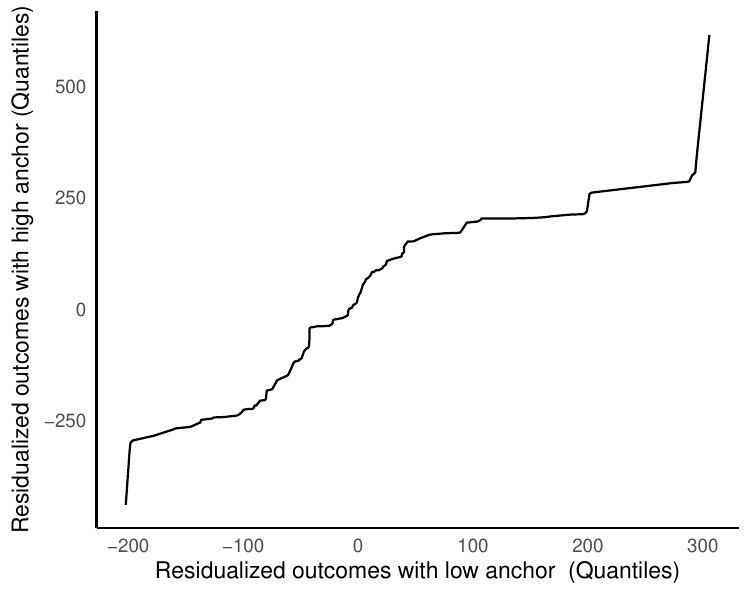}

     \end{subfigure}
     \hfill
        \caption{Q-Q plots for evaluating the gains from applying sharp bounds. (Left) Q-Q plot for the residualized treated and untreated outcomes from \citet{Murrar_18}. (Right) Q-Q plot for the residualized outcomes in the low and high anchoring condition from \cite{Lee_22}. Both Q-Q plots reveal marked deviations from the diagonal, indicating that sharp variance bounds will be smaller than the conventional variance estimator.}
        \label{fig:QQplots_experiments}
\end{figure}
Laboratory experiments in psychology commonly employ designs in which subjects' beliefs and attitudes are measured prior to exposure to the randomly assigned treatment; shortly after exposure to treatment or placebo interventions, subjects' beliefs and attitudes are measured again. Such designs naturally lend themselves to covariate adjustment because the baseline survey contains a highly prognostic measure, the question that is later used to assess outcomes.

Experiment 1 reported in \citet{Murrar_18} is typical of this type of experimental design. The study involves exposing subjects to six episodes of a TV sitcom, either \emph{Little Mosque on the Prairie} (N=87) or \emph{Friends} (N=86). The authors show that these series are similar in many respects, but \emph{Little Mosque on the Prairie} exposes audiences to likable Arab characters, whereas none of the episodes of \emph{Friends} shows any nonwhite characters.  Prior to viewing either the treatment or placebo drama, subjects give their evaluations of several social groups, including Arabs.  The specific question asks, ``Please indicate how warmly or coldly you feel toward the groups listed below...Arab people'' and elicits a rating from 0 to 100.  The same question is asked at endline.  In the placebo group, the correlation between baseline and endline ratings is 0.83. A regression of the post-treatment rating on the pre-treatment rating and a binary treatment indicator scored 1 if subjects were exposed to \emph{Little Mosque on the Prairie} yields an estimated treatment effect of 5.74 with a (conventional) standard error of 1.57.  Randomization inference rejects the sharp null hypothesis of no effect at the $p < 0.001$ level.

In order to get a sense of why sharp bounds might be helpful in this application, we examine the Q-Q plot for the (residualized) treated and untreated outcomes; see Figure \ref{fig:QQplots_experiments} (Left).  The apparent deviation from the 45-degree line is suggestive of heterogeneous effects.  Indeed, when we apply sharp bounds to the covariate-adjusted estimate of the ATE in Table \ref{tab:results_applc_murrar}, we see an appreciable reduction in estimated variance vis-a-vis the conventional estimator and Neyman's Cauchy-Schwarz estimator. The data from \citep{Murrar_18} are publicly available as part of the \texttt{experimentr} R package \citep{Tuncer2022_experimentr}.

\begin{table}[!t]
    \centering
    \begin{tabular}{rcc|cc}
         & Unadjusted & & Adjusted &  \vspace{0.1cm}
         \\\hline 
         & Variance Estimate &Improvement& Variance Estimate& Improvement  \\ 
         && of Sharp vs. && of Sharp vs.\\
         Conventional & 9.29& 2.4\% & 2.38 &15.4\% \\
         Cauchy-Schwarz&9.24 & 1.9\% & 2.30 & 12.3\%\\
         Sharp &9.06 & - & 2.02 & -\\
    \end{tabular}
    \caption{Variance estimates and improvement (in percent) from applying sharp variance bounds to a study of anti-Arab prejudice.  Data are from \citet{Murrar_18}. The left panel reports results for  difference-in-means (unadjusted); the right panel, for linear regression adjustment (adjusted).}
    \label{tab:results_applc_murrar}
\end{table}
\subsection{Behavioral Economics Example: Anchoring Bias}

A longstanding debate in behavioral economics concerns the extent to which people have stable and well-ordered preferences over consumer products, policy outcomes, risk, and other inputs to personal or collective utility.  One challenge to traditional economic models that presuppose stable and well-ordered preferences is the phenomenon called ``anchoring bias'': decisions tend to change systematically depending on the information that people are presented with prior to making a choice.  For example, marketing scholars have noticed that presenting people with initial examples of very expensive consumer goods raises subjects' subsequent willingness to pay for other consumer goods in the same class.  The anchoring phenomenon has been shown not only for consumer preferences, but also for guesses about objective quantities, such as the height of the tallest redwood \citep{Green_Jacowitz_Kahneman_McFadden_1998}.

\citet{Lee_22} present a series of experiments that demonstrate the anchoring effect on consumer preferences concerning Miami hotels. In experiment 3a, 400 subjects completing an online survey were randomly assigned to two conditions, high price or low price.  In each condition, subjects were instructed to imagine ``purchasing a hotel
room for one night during an upcoming trip to Miami'' whereupon they 
were shown the name, a photograph, and a star rating for one of
two hotels. In the low-anchor condition, ``participants first
saw this information and the price of a room for one night
in a one-star Miami Beach hotel...In the high-anchor
condition, participants saw this information and the price
of a room for one night in a five-star Miami Beach hotel.'' (p.67)  Subjects in both conditions were later shown analogous information for ``either a two-star or a four-star Miami Beach hotel...without prices'';  the outcome measure is the maximum amount each subject reports being willing to pay for a room.  The anchoring hypothesis is that the higher anchor will generate a higher expressed willingness to pay.  A regression of willingness to pay for the specified hotel on treatment assignment (a description of either a one-star or a five-star hotel) and covariates (whether the target hotel is described as two-star or four-star, as well as subjects' gender and age) yields a large anchoring effect of \$101.3 with a (conventional) standard error of \$11.1.  Their data are publicly available \citep{Lee2025_noiseAnchoring}.

The right pane of Figure \ref{fig:QQplots_experiments} shows the Q-Q plot for the residualized outcomes in the low and high anchoring conditions. The deviation from a straight-line relationship suggests that sharp bounds will render a smaller variance estimate than the conventional variance estimator.
Indeed, Table \ref{tab:results_applc_lee} confirms that sharp variance bounds offer substantial improvements over both the conventional and Neyman's Cauchy-Schwarz variance estimates. 
\begin{table}[!t]
    \centering
    \begin{tabular}{rcc|cc}
         & Unadjusted & & Adjusted &  \vspace{0.1cm}
         \\\hline 
         & Variance Estimate &Improvement& Variance Estimate& Improvement  \\ 
         && of Sharp vs. && of Sharp vs.\\
         Conventional & 176.5& 14.6\% &122.3 &12.1\% \\
         Cauchy-Schwarz&168.4 & 10.5\% & 119.2 & 9.8\%\\
         Sharp &150.7 & - & 107.5 & -\\
    \end{tabular}
    \caption{Variance estimates and improvement (in percent) from applying sharp variance bounds to a study of willingness to pay for hotel stays.  Data are from \citet{Lee_22}. The left panel reports results for  difference-in-means (unadjusted); the right panel, for linear regression adjustment (adjusted).}
    \label{tab:results_applc_lee}
\end{table}

\subsection{Biomedical Example: Two Fluid-Management Strategies in Acute Lung Injury}
\begin{table}[!t]
    \centering
    \begin{tabular}{rcc|cc}
         & Unadjusted & & Adjusted &  \vspace{0.1cm}
         \\\hline 
         & Variance Estimate &Improvement& Variance Estimate& Improvement  \\ 
         && of Sharp vs. && of Sharp vs.\\
         Conventional & 0.00929& 3.7\% &0.00865 & 1.7\% \\
         Cauchy-Schwarz&0.00925 & 3.3\% & 0.00863 & 1.4\%\\
         Sharp &0.00895 & - & 0.00850& -\\
    \end{tabular}
    \caption{Variance estimates and improvement (in percent) from applying sharp variance bounds to a study of management of acute lung injury.  Data are from \citet{BioExample}. The left panel reports results for  difference-in-means (unadjusted); the right panel, for linear regression adjustment (adjusted).}
    \label{tab:results_bio}
\end{table}
\begin{figure}[!t]
     \centering
     \begin{subfigure}[b]{0.49\textwidth}
         \centering
    \includegraphics[width=\textwidth]{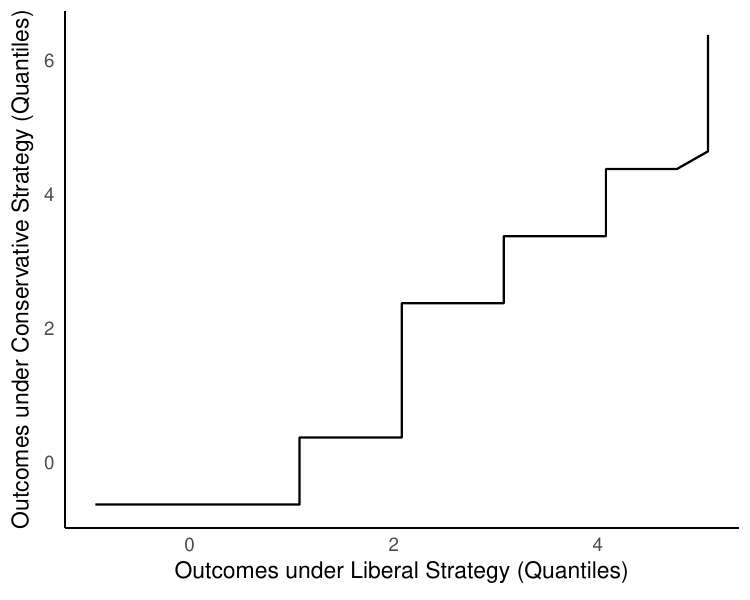}

     \end{subfigure}
     \hfill
     \begin{subfigure}[b]{0.49\textwidth}
         \centering
         \includegraphics[width=\textwidth]{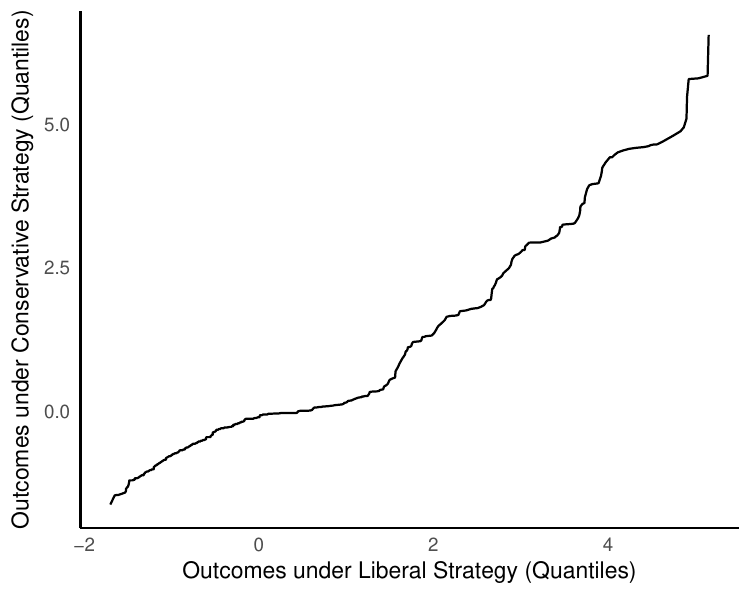}

     \end{subfigure}
     \hfill
        \caption{(Left) Q-Q plot for the outcomes (number of ICU-free days) under a liberal and conservative treatment strategy from \citep{BioExample}. (Right) Q-Q plot for the residualized outcomes adjusting for the patients' APACHE II score (measure of patient's severity-of-disease). Both Q-Q plots reveal deviations from the diagonal, but these deviations are fairly limited after covariate adjustment, indicating that sharp variance bounds will be only slightly smaller than the conventional variance estimator.}
        \label{fig:QQplots_bio}
\end{figure}
Optimal fluid management strategies for patients with acute lung injury have been an active area of experimental investigation. \citet{BioExample} performed a randomized experiment comparing two such strategies (which they refer to as  conservative and liberal) and measured their effects on multiple medical endpoints. We focus on the effect of these strategies on the number of ICU-free days during the first 7 days of treatment, an outcome for which the original study found a significant difference between the two strategies. The Fluid and Catheter Treatment Trial (FACTT) data are available from BioLINCC \citep{FACTT2006}.

Figure \ref{fig:QQplots_bio} (Left) shows the Q-Q plot for this outcome under the liberal and conservative strategy; the deviation from the diagonal indicates benefits of using our sharp variance bounds.
For this outcome, we find that our sharp variance bound improves variance estimation by 3.3\% compared to Neyman's Cauchy-Schwarz bound (see Table \ref{tab:results_bio}). Although the original authors do not use regression adjustment, they collect the patients' APACHE II scores at baseline. This score is a measure of a patient's severity-of-disease and is (significantly) predictive of time spent in the ICU.\footnote{35 out of 1000 patients did not have an APACHE II score. We dropped these patients from all analyses for the sake of comparability to the unadjusted results. Removing these observations does not materially change the originally reported results.} Adjusting for the APACHE II score increases the precision of the ATE.
Figure \ref{fig:QQplots_bio} (Right) shows the Q-Q plot for the residualized outcomes. While residualization smooths the Q-Q plot, suggesting slightly diminished benefits from our sharp variance bounds compared to the unadjusted case, the plot still shows some deviation from the diagonal.
Using regression-adjustment based on the APACHE II score, our sharp variance bounds improve precision by 1.4\% compared to Neyman's Cauchy-Schwarz bounds; see Table \ref{tab:results_bio}. We find comparable benefits of using our sharp variance bounds (with and without regression adjustment) for other endpoints reported by the study, including their primary endpoint (60-day mortality rate) and 28-day ventilator-free days.

%% file: paper/conclusion.tex
\section{Conclusion}
This paper derives sharp variance bounds for general regression adjustment for two-arm randomized experiments. We provide a consistent estimator for these sharp variance bounds. For illustration, we apply our sharp variance estimators to linear regression adjustment \citep{Lin_reg_adj} and a decorrelation method for general regression adjustment \citep{Su_Mou_Ding_Wainwright_2023}. Moreover, we provide software (in the form of an R package that enables experimental researchers to make use of this method.

Our sharp variance estimator yields less conservative variance estimates for randomized experiments than conventional variance estimators. This improved method of variance estimation has practical implications for power calculations, confidence intervals, and significance testing in randomized experiments.

That said, sharp variance bounds may provide negligible improvements when the distributions of the adjusted potential outcomes in treatment and control are similar.  In many experimental applications, treatment effects appear to be small and homogeneous, in which case sharp bounds will generate variance estimates that are similar to (but not larger than) conventional estimates.
By the same token, the benefits of sharp variance estimates may be marginal in applications where almost all of the variation in outcomes is predicted by the observed covariates, as this reduces the heterogeneity in the adjusted outcomes that drives the improvement beyond Neyman's Cauchy-Schwarz estimator.

Less conservative variance estimators, and sharp estimators in particular, have the potential for large cumulative benefits even if the gains in precision may be modest for individual studies. 
The potential advantages of sharp variance bounds are especially relevant for biomedical trials, which often involve small numbers of subjects but evaluate treatments that are expected ex ante to produce heterogeneous effects. Such trials frequently use regression adjustment to wring additional precision from limited data, and the method presented here offers an opportunity to estimate the variance in a less conservative manner.

Despite their advantages, these sharp variance bounds remain limited to simple experimental designs. It is unclear how to extend the notion of sharp variance bounds to more complicated designs in which the variance of the treatment effect estimator may not necessarily be expressible purely in terms of a joint distribution.

%% file: paper/appendix.tex
\section*{Appendix:  Sharp Bounds on the Variance of General Regression Adjustment in Randomized Experiments}
\appendix
\section{Sharp Variance Bounds for Bernoulli Random Experiments}
Similar to the discussion in the main text, in a Bernoulli Random Experiment (BRE), the variance of the effect estimator remains unidentified. The variance under this design is given by
\begin{align}\label{Eq:VarOracleBRE}
      \text{Var}(\hat \tau_N(f_1,f_0)) = \frac{1}{N} \bigg(\frac{n_0}{n_1} \E_{G_N}\big[\varepsilon(1)^2\big] + \frac{n_1}{n_0} \E_{F_N}\big[\varepsilon(0)^2\big] + 
         2 \E_{\Gamma_N}\big[\varepsilon(1)\varepsilon(0)\big]\bigg).
\end{align}
We can immediately see that the discussion of the main paper applies similarly to BREs.
In this case, the sharp variance bound takes the form
\begin{align}
\label{Eq:SVB_BRE}
     \text{SVB}_N & = \frac{1}{N} \bigg( \frac{n_0}{n_1} \E_{G_N}\big[\varepsilon(1)^2\big] + \frac{n_1}{n_0} \E_{F_N}\big[\varepsilon(0)^2\big] + 
         2 \E_{\Gamma^U_N}\big[\varepsilon(1)\varepsilon(0)\big]\bigg),
\end{align}
with the joint of extremal dependence $\Gamma^U_N(\xi_1,\xi_0) = \min\{G_N(\xi_1),F_N(\xi_0)\}$.

To estimate the sharp variance bound, we work with the empirical CDFs $\hat G_N(\xi) = 1/{\tilde{n}_1} \sum_{i=1}^N Z_i \mathbbm{1}\{\hat \varepsilon_{i}(1) \leq \xi\}$ and $\hat F_N(\xi) = 1/({N-\tilde {n}_1} )\sum_{i=1}^N (1-Z_i)\mathbbm{1}\{\hat \varepsilon_{i}(0) \leq \xi\}$, where $\tilde n_1 = \sum_{i=1}^N Z_i$. In case that either $\tilde{n}_1$ or $N- \tilde{n}_1$ are 0, we define the respective empirical measure to be $\delta(0)$. The probability of this happening vanishes as $N$ tends to infinity so that this choice does not affect our asymptotic results. For details, see the proof of Proposition \ref{Prop:ConsistencyOfSharpBounds_star} in Appendix \ref{Appendix:proof_prop1}.  These yield the plug-in estimators for the sharp variance bound
\begin{align}
\label{Eq:estmSVB_BRE}
     \widehat{\text{SVB}}_N =\frac{1}{N} \bigg( \frac{n_0}{n_1} \E_{\hat G_N}\big[\hat \varepsilon(1)^2\big] + \frac{n_1}{n_0} \E_{\hat F_N}\big[\hat \varepsilon(0)^2\big] + 
         2 \E_{\hat \Gamma^U_N}\big[\hat \varepsilon(1)\hat \varepsilon(0)\big]\bigg) 
\end{align}
with $\hat\Gamma^U_N(\xi_1,\xi_0) = \min\{\hat G_N(\xi_1),\hat F_N(\xi_0)\}$.
\section{Theoretical Results and Proof of Proposition \ref{Prop:ConsistencyOfSharpBounds}}
\label{App:Theory}
To make our discussion of our theoretical results more general, we will consider the general regression-adjusted estimator
\begin{align}
\label{eq:reg_adj_general}
    \hat \tau_N = \frac{1}{n_T} \sum_{i=1}^N T_i(Y_i(1) - f_1(X_i)) - \frac{1}{n_{\bar T}} \sum_{i=1}^N \bar{T}_i(Y_i(0) - f_0(X_i)) + \frac{1}{N} \sum_{i=1}^N (f_1(X_i) - f_0(X_i)),
\end{align}
based on general sampling indicators $T_i$ and $\bar{T}_i$ that determine which units are included in the calculation of the adjusted estimator. They need to satisfy $T_i \leq Z_i$ and $\bar T_i \leq 1-Z_i$ because we can only include observed potential outcomes.
These indicators can be chosen to be the treatment indicator $(Z_i,1-Z_i)$ themselves so that all units are included in the calculation of the adjusted estimator. However, they can also be chosen to fulfill other properties (such as decorrelation, see \citep{Su_Mou_Ding_Wainwright_2023}). To accommodate both the setting for BREs and CREs, we will make the following assumption about the indicators: 
\begin{assumption}[Random indicators]
\label{as:RandomAssignment}
    The random indicators $(T_i,\bar{T}_i)$ are one of the following: 
    \begin{enumerate}
        \item Bernoulli random variables independent across units such that $\text{Cov}(T_i,T_j) = \pi_T (1-\pi_T) \delta_{ij}$, $\text{Cov}(\bar{T}_i,\bar{T}_j) = \pi_{\bar{T}} (1-\pi_{\bar{T}}) \delta_{ij}$ and  $\text{Cov}(T_i,\bar{T}_j) = - \pi_T \pi_{\bar{T}} \delta_{ij}$. Assume $\pi_T, \pi_{\bar T} \in (0,1)$ and define $ n_T := N\pi_T$ and $n_{\bar T} :=N \pi_{\bar{T}}$.
        \item$(T_i,\bar{T}_i) = (T_i, 1-T_i)$, where $T_i$ encodes a simple random sample (without replacement) of $n_T$ out of $N$ units. In this case, further assume that $n_T/N \rightarrow \pi_T \in (0,1)$ as $N \rightarrow \infty$ and define $n_{\bar{T}} = N - n_T$ and $\pi_{\bar{{T}}} = 1- \pi_T$.
    \end{enumerate}
\end{assumption}

In this section, we will prove a more general version of Proposition \ref{Prop:ConsistencyOfSharpBounds} as presented in the main text. We begin by appropriately defining the sharp variance bound $\text{SVB}$ and its estimator $\widehat{\text{SVB}}$ for the general regression-adjusted estimator of this section (Equation \ref{eq:reg_adj_general}) for both BREs and CREs:
We have $\hat G_N (\xi) = \frac{1}{\tilde n_T}\sum_{i=1}^N T_i \mathbbm{1}\{\hat \varepsilon_{i,N}(1) \leq \xi\}$ and   $\hat F_N(\xi) = \frac{1}{\tilde n_{\bar T}}\sum_{i=1}^N \bar{T}_i \mathbbm{1}\{\hat \varepsilon_{i,N}(0) \leq \xi\}$ with  $\tilde n_T = \sum_{i=1}^N T_i$ and $\tilde n_{\bar T} = \sum_{i=1}^N \bar T_i$.\footnote{In case that either $\tilde n_T$ or $\tilde n_{\bar T}$ is 0, we define the respective empirical distribution to be $\delta(0)$. The probability of this happening vanishes as $N$ tends to infinity so that this choice does not affect our asymptotic results. For details, see the proof of Proposition \ref{Prop:ConsistencyOfSharpBounds_star} in Appendix \ref{Appendix:proof_prop1}. }
Additionally, define $\hat \Gamma^U_N(\xi_1,\xi_0) = \min\{\hat G_N(\xi_1),\hat F_N(\xi_0)\}$.

Under Assumption \ref{as:RandomAssignment}.1 (BRE), we have 
\begin{align*}
     \text{SVB}_N & = \frac{1}{N} \bigg( \frac{N-n_T}{n_T} \E_{G_N}\big[\varepsilon(1)^2\big] + \frac{N-n_{\bar T}}{n_{\bar T}} \E_{F_N}\big[\varepsilon(0)^2\big] + 
         2 \E_{\Gamma^U_N}\big[\varepsilon(1)\varepsilon(0)\big]\bigg) \enspace \text{and} \\
        \widehat{\text{SVB}}_N & = \frac{1}{N} \bigg( \frac{N-n_T}{n_T} \E_{\hat G_N}\big[\hat \varepsilon(1)^2\big] + \frac{N-n_{\bar T}}{n_{\bar T}} \E_{\hat F_N}\big[\hat \varepsilon(0)^2\big] + 
         2 \E_{\hat \Gamma^U_N}\big[\hat \varepsilon(1)\hat \varepsilon(0)\big]\bigg),
\end{align*}
which generalize Equation \ref{Eq:SVB_BRE} and \ref{Eq:estmSVB_BRE}.

Similarly, under Assumption \ref{as:RandomAssignment}.2 (CRE), we have 
\begin{align*}
     \text{SVB}_N & = \frac{1}{N-1} \bigg( \frac{N-n_T}{n_T} \text{Var}_{G_N}\big(\varepsilon(1)\big) + \frac{N-n_{\bar T}}{n_{\bar T}} \text{Var}_{F_N}\big(\varepsilon(0)\big)  
       +  2 \text{Cov}_{\Gamma^U_N}\big(\varepsilon(1),\varepsilon(0)\big)\bigg) \enspace \text{and} \\
        \widehat{\text{SVB}}_N & = \frac{1}{N-1} \bigg( \frac{N-n_T}{n_T} \text{Var}_{\hat G_N}\big(\hat \varepsilon(1)\big) + \frac{N-n_{\bar T}}{n_{\bar T}} \text{Var}_{\hat F_N}\big(\hat \varepsilon(0)\big)  
       +  2 \text{Cov}_{\hat \Gamma^U_N}\big(\hat \varepsilon(1),\hat \varepsilon(0)\big)\bigg),
\end{align*}
which generalize Equation
\ref{Eq:SVB_CRE} and \ref{Eq:estmVarBoundsReg}.

With this notation, we can formulate a more general version of Proposition \ref{Prop:ConsistencyOfSharpBounds}:
\pushQED{\qed} 
\renewcommand{\theproposition}{1*}

\begin{proposition}
\label{Prop:ConsistencyOfSharpBounds_star}
 Let Assumptions \ref{as:convergJoint} and \ref{as:RandomAssignment} hold. Further, assume:
     \begin{enumerate}[label=(\alph*)]
    
         \item The estimates of the outcome models are $o_p(1)$ consistent, that is,
    \begin{align}\nonumber
       \big( \frac{1}{N} \sum_{i=1}^N (f_{q,N}(X_i) - \hat f_{q,N}(X_i))^2 \big)^{1/2}= o_p(1) \hspace{0.3cm} \text{for } q \in \{0,1\}.
    \end{align} 
         \item The population-adjusted potential outcomes $\varepsilon_{i,N}(q) = Y_i(q) - f_{q,N}(X_i)$ are uniformly square-integrable, that is
         \begin{align}\nonumber
              \sup_N \frac{1}{N} \sum_{i=1}^N \varepsilon_{i,N}(q)^2 \mathbbm{1}[\varepsilon_{i,N}(q)^2 \geq \beta] \rightarrow 0
         \end{align}
        as $\beta \rightarrow \infty$ for $q \in \{0,1\}$.
     \end{enumerate}
    Then
     \[
     \widehat {\emph{\text{SVB}}}_N - \emph{SVB}_N =o_p(1/N).
     \]
 \end{proposition}
 \renewcommand{\theproposition}{\arabic{proposition}}
Under Assumption \ref{as:RandomAssignment}.2, and setting $(T_i,\bar T_i) =(Z_i,1-Z_i)$, we recover Proposition \ref{Prop:ConsistencyOfSharpBounds} as stated in the main text.

We will first derive some results on empirical process theory for finite populations. Then, we will proceed by giving a proof of Proposition \ref{Prop:ConsistencyOfSharpBounds_star}.
\subsection{Some Results on Empirical Process Theory for Finite Populations}
\label{App:EmpProc}
In this section, we will extend some results of empirical process theory to the finite-population setting. 
Let $\{T_i\}_{i=1}^N$ be random (sampling) indicators that fulfill Assumption \ref{as:RandomAssignment},  i.e., $\{T_i\}_{i=1}^N$ are either iid. Bernoulli random variables with $P(T_i = 1) = \pi_T$ (in this case, define $n_T:=\pi_T N$) or they encode a simple random sample without replacement in which $n_T$ units are sampled from a population of size $N$. Based on these random sampling indicators, we have $\hat P_N = \frac{1}{n_T} \sum_{i=1}^N T_i \delta_{y_i} $ as an (Horvitz-Thompson) approximation to the empirical distribution $ P_N = \frac{1}{N} \sum_{i=1}^N \delta_{y_i} $ of the population $\{y_i\}_{i=1}^N$.

We are interested in convergence results for 
\begin{align}\nonumber
   \lVert \hat P_N - P_N \rVert_{\mathcal{F}} \, &= \, \sup_{f \in \mathcal{F}} \frac{1}{n_T}\sum_{i=1}^N T_i f(y_i) - \frac{1}{N}\sum_{i=1}^N f(y_i) ,
\end{align}
where $\mathcal{F}$ is some class of measurable functions.\footnote{To simplify our presentation, we will refrain from using outer probabilities. To avoid potential measurability issues, we will instead assume that $\mathcal{F}$ is pointwise measurable, see \citep{Vaart_Wellner_1996}.} While $\hat P_N$ is neither a probability distribution nor a (classical) empirical distribution, these Horvitz-Thompson-type empirical distributions \citep{Bertail_Chautru_Clémençon_2017} lend themselves to propositions similar to those in empirical process theory. Our goal here is to derive results on which classes of functions $\mathcal{F}$ are P-Glivenko-Cantelli (see Proposition \ref{Prop:GC}). 

We will require the following two lemmata provided by \citet{Hoeffding_1963}.
The first is the well-known result for bounded i.i.d. random variables.
\begin{lemma}[Hoeffding Inequality for Bernoulli design]
\label{lemma:hoeffding-iid}
Let $X_1,...,X_n$ be independent with $a_i \leq X_i \leq b_i$ a.s. (i=1,2,...,n), then for $t>0$
\begin{align}
    P(\big\lvert \frac{1}{n}\sum_{i=1}^n X_i - \mathbb{E}[\frac{1}{n}\sum_{i=1}^n X_i] \big\rvert \geq t) \leq 2\exp(-2\frac{n^2t^2}{\sum_{i=1}^n(b_i - a_i)^2}).
\end{align}
\end{lemma}
Hoeffding extended this result to simple random samples without replacement \citep{Hoeffding_1963,Bardenet_Maillard_2015,Bertail_Clémençon_2019}:
\begin{lemma}[Hoeffding Inequality for simple random samples without replacement]
\label{lemma:hoeffding-srswr}
Let $\mathcal{Y} = \{y_1,...,y_N\}$ be a finite population of $N$ points and let $X_1,...,X_n$ be a simple random sample (without replacement) drawn from $\mathcal{Y}$. Let
\begin{align}
    a = \min_{1\leq i \leq N} y_i \hspace{0.5cm }\text{and}\hspace{0.5cm } b = \max_{1\leq i \leq N} y_i
\end{align}

Then for all $t>0$ 
\begin{align}
    P(\big\lvert \frac{1}{n}\sum_{i=1}^n X_i - \frac{1}{N}\sum_{i=1}^N y_i \big\rvert \geq t) \leq 2 \exp(-2\frac{nt^2}{(b-a)^2}).
\end{align}
\end{lemma}

We will now provide a proposition giving a sufficient condition for a class of functions $\mathcal{F}$ to be Glivenko-Cantelli in probability.
These lemmata are useful because they allow us to bound the Orlicz norm $\lVert X\rVert_{\psi_p} = \inf\{C>0 : \mathbb{E}[\psi_p(\frac{\lvert X\rvert}{C}])\leq 1\}$, where $\psi_p(x) = \exp(x^p) -1$ for $p\geq 1$. These Orlicz norms, in turn, bound the $L_p$-norms, $\lVert X \rVert_p \leq \lVert X \rVert_{\psi_p}$.

The following Lemma \citep[Lemma 2.2.1]{Vaart_Wellner_1996} gives bounds of the $\psi_p$-Orlicz-norm based on tail inequalities.
\begin{lemma}
\label{lemma:orlicz}
    Let $X$ be a random variable with $P(\lvert X \rvert > x) \leq K \exp(-Cx^p)$ for every $x$, for constants $K$ and $C$, and for $p\geq1$. Then its Orlicz norm is bounded by
    \begin{align}
        \lVert X \rVert_{\psi_p} \leq \bigg(\frac{1+K}{C}\bigg)^{1/p}.
    \end{align}
\end{lemma}

Bounds on the Orlicz norm, in turn, provide maximal inequalities. We will require the following lemma \citep[Lemma 2.2.2]{Vaart_Wellner_1996}:
\begin{lemma}
\label{lemma:max-ineq}
    For any random variables $X_1,...,X_m$, we have
    \begin{align}
        \lVert\max_{1\leq i \leq m} X_i \rVert_{\psi_p} \leq K \psi_p^{-1}(m) \max_{1\leq i \leq m} \lVert X_i \rVert_{\psi_p},
    \end{align}
where $K$ is a constant only depending on $\psi_p$.
\end{lemma}

We will now provide a proposition giving a sufficient condition for a class of functions $\mathcal{F}$ to be Glivenko-Cantelli in probability.
\begin{proposition}[P-Glivenko-Cantelli by entropy]
\label{Prop:GC}
Let $\mathcal{F}$ be a class of measurable functions with envelope $F$ such that $P_NF< \infty$. Let $\mathcal{F}_M$ be the class of functions $f\mathbbm{1}[F\leq M]$ for all $f \in \mathcal{F}$. 
Suppose 
\begin{align}\nonumber
    \frac{1}{N} \log N(\varepsilon,\mathcal{F}_M,L_1( P_N)) 
   \rightarrow 0
\end{align}
for every fixed $M < \infty$ and every $\varepsilon >0$, and suppose the envelope tails are uniformly negligible
\[
    \lim_{M\to\infty}\limsup_{N\to\infty}
    P_N\big[F\mathbbm{1}[F>M]\big]
    =
    0.
\]
Then under Assumption $\ref{as:RandomAssignment}$

\begin{align}\nonumber
    \lVert \hat P_N - P_N \rVert_{\mathcal{F}} = \sup_{f \in \mathcal{F}} \lvert \frac{1}{n_T}\sum_{i=1}^N T_i f(y_i) - \frac{1}{N} \sum_{i=1}^N f(y_i) \rvert \rightarrow 0 \hspace{0.5cm} \text{in probability}.
\end{align}.
\end{proposition}
\begin{proof}
    By the triangle inequality, we have
\begin{align}\nonumber
    \mathbb{E}\lVert \hat P_N - P_N \rVert_{\mathcal{F}} \leq \mathbb{E}\bigg\lVert \frac{1}{n_T} \sum_{i=1}^N T_i f(y_i) - \frac{1}{N}\sum_{i=1}^Nf(y_i)\bigg\rVert_{\mathcal{F}_M} + 2P_N\big[F\mathbbm{1}[F>M]\big]
\end{align}
for every $M$. 

We will begin by showing that the first term converges to zero for a fixed $M$. Let $\mathcal{G}$ be an $\eta$-net in $L_1(P_N)$ over $\mathcal{F}_M$, then for any $f \in \mathcal{F}_M$, there exists a $g \in \mathcal{G}$ such that
\begin{align}\nonumber
    \bigg\lvert\frac{1}{n_T} \sum_{i=1}^N T_i f(y_i) - \frac{1}{N}\sum_{i=1}^Nf(y_i) \bigg\rvert \, &\leq \, \bigg\lvert\frac{1}{n_T} \sum_{i=1}^N T_i g(y_i) - \frac{1}{N}\sum_{i=1}^Ng(y_i) \bigg\rvert + \bigg\lvert (\hat P_N - P_N) \big(f -g\big) \bigg\rvert \\
    \, &\leq \, \bigg\lVert \frac{1}{n_T} \sum_{i=1}^N T_i g(y_i) - \frac{1}{N}\sum_{i=1}^N g(y_i) \bigg\rVert_{\mathcal{G}} +\big(\frac{N}{n_T}+1\big)\eta.
\end{align}
Consequently,
\begin{align}
\label{eq:L1}
   \mathbb{E}\bigg\lVert \frac{1}{n_T} \sum_{i=1}^N T_i f(y_i) - \frac{1}{N}\sum_{i=1}^Nf(y_i)\bigg\rVert_{\mathcal{F}_M} \leq  \mathbb{E}\bigg\lVert \frac{1}{n_T} \sum_{i=1}^N T_i g(y_i) - \frac{1}{N}\sum_{i=1}^N g(y_i) \bigg\rVert_{\mathcal{G}} +\big(\frac{N}{n_T}+1\big)\eta.
\end{align}
The cardinality of $\mathcal{G}$ can be chosen to be $N(\eta,\mathcal{F}_M,L_1(P_N))$. 
Under Assumption \ref{as:RandomAssignment} and Lemma \ref{lemma:hoeffding-iid} and \ref{lemma:hoeffding-srswr}, we have
\begin{align} \nonumber
    P(\lvert \frac{1}{n_T} \sum_{i=1}^N T_i g(y_i) - \frac{1}{N}\sum_{i=1}^N g(y_i) \rvert \geq t) \leq 2 \exp\bigg(- c\frac{n_T t^2}{M^2}\bigg) 
\end{align}
where $c>0$ is some constant depending on the sampling design (among Assumption \ref{as:RandomAssignment}) and where we have used that $-M \leq g(y_i) \leq M$ (after truncating $g$, if necessary). Consequently, Lemma \ref{lemma:orlicz} yields that
\begin{align}
    \lVert \frac{1}{n_T} \sum_{i=1}^N T_i g(y_i) - \frac{1}{N}\sum_{i=1}^N g(y_i) \rVert_{\psi_2} \leq C\frac{M}{\sqrt{n_T}},
\end{align}
where $C < \infty $ (only depends on $c$).
Bounding the $L_1$-norm on the right-hand side of Equation \ref{eq:L1} by the $\psi_2$-Orlicz-norm, application of Lemma \ref{lemma:max-ineq} now yields
\begin{align}
\label{eq:entropy-bound}
    \mathbb{E}\bigg\lVert \frac{1}{n_T} \sum_{i=1}^N T_i f(y_i) - \frac{1}{N}\sum_{i=1}^Nf(y_i)\bigg\rVert_{\mathcal{F}_M} \leq K C\frac{M}{\sqrt{n_T}} \sqrt{1+\log N(\eta,\mathcal{F}_M,L_1(P_N))}+\big(\frac{N}{n_T}+1\big)\eta,
\end{align}
where $K$ is a universal constant. Under Assumption \ref{as:RandomAssignment}, we have $\frac{n_T}{N} \rightarrow \pi_T \in (0,1)$. Moreover,
 $\sqrt{\log N(\eta,\mathcal{F}_M,L_1(P_N)}\frac{1}{\sqrt{N}}$ tends to zero by assumption, hence the right side of Equation \ref{eq:entropy-bound} tends to $(1/\pi_T+1)\eta$. The argument is valid for every $\eta >0 $ so that we can conclude that the left side of Equation \ref{eq:entropy-bound} converges to zero.

Finally, 
\[
\limsup_{N \rightarrow \infty} \E \lVert \hat P_N - P_N \rVert_{\mathcal{F}} \leq 2 \limsup_{N \rightarrow \infty} P_N[F\mathbbm{1}[F > M]].
\]
Letting $M \rightarrow \infty$ and using the uniform tail condition, we have $\lVert \hat P_N - P_N \rVert_{\mathcal{F}} \rightarrow 0$ in mean and hence in probability.

\end{proof}

The proof of Proposition \ref{Prop:ConsistencyOfSharpBounds_star} will require the estimates $\hat G_N$ and $\hat F_N$ of the marginals in the finite-population setting to converge weakly to limits $G$ and $F$.
To show this, we will use the bounded-Lipschitz distance. We will require the following Lemma bounding the entropy of this class of functions:
\begin{lemma}
\label{lemma:BL}
    Let $\text{BL} := \{f: \mathbb{R} \rightarrow [-1,1]| f \text{ is } 1\text{-Lipschitz}\}$ be the class of all 1-bounded-Lipschitz functions. 
    Let $P_N$ be such that $\mathbb{E}_{P_N}[\lvert Y \rvert] = \frac{1}{N}\sum_{i=1}^N \lvert y_i \rvert \leq C$ for some constant $C$.
    Then
    \begin{align}
        \log N(\varepsilon,\text{BL},L_1(P_N)) \leq A \frac{1+C}{\varepsilon^2}
    \end{align}
    for some constant $A$ and for all $\varepsilon >0$.
\end{lemma}
\begin{proof}
    For $\varepsilon >1$ take $f_0 \equiv 0$ and observe that for any $f \in \text{BL}$, we have $\lVert f- f_0\rVert_{L_1(P_N)} \leq 1 < \varepsilon$ and hence $N(\varepsilon,\text{BL},L_1(P_N)) = 1$.

    Let $0 < \varepsilon < 1$. We will construct an $\varepsilon$-cover of BL (under the $L_1(P_N)$-norm) with cardinality less than $\exp(A\frac{C}{\varepsilon^2 }) $ for some $A>0$. This will complete the proof as $N(\varepsilon,\text{BL},L_1(P_N))$ will then be bounded by $\exp(A\frac{C}{\varepsilon^2 })$.

    Define $B \in \mathbb{R}$ such that 
    \begin{align}
    \label{eq:gridbound}
        \frac{1}{N} \sum_{i=1}^N \mathbbm{1}\{\lvert y_i \rvert \geq B\} \leq \varepsilon.
    \end{align}
    We have $\mathbb{E}_{P_N}[\mathbbm{1}\{\lvert y_i \rvert \geq B\}] \leq \frac{C}{B}$ by Markov's inequality and the first-moment bound. Choosing $B = \lfloor \frac{C}{\varepsilon}\rfloor + 1$ thus satisfies Equation \ref{eq:gridbound}.
     We now construct an $\varepsilon$-grid covering the interval $[-B,B]$: Let $a_k := k \varepsilon - B$ for $k=0,...,M$, where $M=\lceil \frac{2B}{\varepsilon}\rceil$.
     Additionally, define $B_k := (a_{k-1},a_k]$ for $k=1,...,M$. For each $f \in BL$, define 
    \begin{align}
        \tilde f(x) = \sum_{k=1}^M q_k(f) \mathbbm{1}_{B_k}(x).
    \end{align}
with $q_k(f):= -1+\varepsilon
\left\lfloor\frac{f(a_k)+1}{\varepsilon}\right\rfloor$ for $k=1,\ldots,M$.
Note that
  \[
|f(a_k)-q_k(f)|\leq \varepsilon.
\]
Therefore, for $x\in B_k$,
\[
|f(x)-\tilde f(x)|
\le |f(x)-f(a_k)|+|f(a_k)-q_k(f)|
\le 2\varepsilon
\]
    where we have used that $f$ is 1-Lipschitz. 
    For values outside of the grid over $[-B,B]$, we have
    \begin{align}
        \frac{1}{N}\sum_{i: \lvert y_i \rvert \geq B} \lvert f(y_i) - \tilde f(y_i) \rvert \leq \frac{2}{N} \sum_{i=1}^N \mathbbm{1}\{\lvert y_i \rvert \geq B\} \leq 2\varepsilon,
    \end{align}
    where in the first inequality we have used that $f$ is bounded.
    Hence $\lVert f-\tilde f \rVert_{L_1(P_N)} \leq 4 \varepsilon$.

    We now want to determine the cardinality of the set $\{\tilde f : f\in \text{BL}\}$. As $f$ varies over $\text{BL}$, there are at most $\lfloor \frac{2}{\varepsilon}\rfloor +1$ choices for $\tilde f(a_1)$ (because $f$ is bounded to be in $[-1,1]$).
    Note that for any $\tilde f$ and  $k \in \{2 ,..., M\}$, we have
    \begin{align}
        \lvert \tilde f(a_k) - \tilde f(a_{k-1})\rvert \leq  \lvert \tilde f(a_k) - f(a_{k})\rvert  +  \lvert \ f(a_k) -  f(a_{k-1})\rvert  +  \lvert  f(a_{k-1}) - \tilde f(a_{k-1})\rvert \leq 3 \varepsilon.
    \end{align}
Therefore there are at most $7$ choices for $\tilde f(a_k)$ once $\tilde f(a_{k-1})$ has been chosen.

The collection $\{\tilde f : f\in \text{BL}\}$ is thus a $4\varepsilon$-cover (w.r.t. the $L_1(P_N)$-norm) of BL and has a cardinality upper bounded by  $(\lfloor \frac{2}{\varepsilon}\rfloor +1 ) 7^{M}$. Hence
\begin{align}
    N(4\varepsilon,\text{BL},L_1(P_N)) \leq \bigg(2\lfloor \frac{2}{\varepsilon}\rfloor +1 \bigg) 7^{\lceil \frac{2(\lfloor C / \varepsilon \rfloor +1)}{\varepsilon} \rceil},
\end{align}
which completes the proof.
\end{proof}

\section{Proof of Proposition \ref{Prop:ConsistencyOfSharpBounds_star}}
\label{Appendix:proof_prop1}

\begin{proof}
We have $\hat G_N (\xi) = \frac{1}{\tilde n_T}\sum_{i=1}^N T_i \mathbbm{1}\{\hat \varepsilon_{i,N}(1) \leq \xi\}$ and   $\hat F_N(\xi) = \frac{1}{\tilde n_{\bar T}}\sum_{i=1}^N \bar{T}_i \mathbbm{1}\{\hat \varepsilon_{i,N}(0) \leq \xi\}$. Under Assumption \ref{as:RandomAssignment}.1, $\tilde n_T$ and $\tilde n_{\bar T}$ are binomial random variables, whereas under Assumption \ref{as:RandomAssignment}.2, they are fixed to $\tilde n_T = n_T$ and $\tilde n_{\bar T} = n_{\bar T}$. In case either $\tilde n_T$ or $\tilde n_{\bar T}$ are 0, we define $\hat G_N = \delta(0)$ or  $\hat F_N = \delta(0)$, respectively. Let $\mathcal{D}$ be the event that $\tilde n_T > 0 $ and $\bar{\mathcal{D}}$ the event that $\tilde n_{\bar T} > 0 $.
We then have the following (random) expectations $\E_{\hat G_N}[\hat \varepsilon_{N}(1)]  = \frac{\mathbbm{1}_\mathcal{D}}{\tilde n_T} \sum_{i=1}^N T_i \hat \varepsilon_{i,N}(1)$ and  $\E_{\hat F_N}[\hat \varepsilon_{N}(0)] = \frac{\mathbbm{1}_\mathcal{\bar D}}{\tilde n_{\bar T}} \sum_{i=1}^N \bar{T}_i \hat \varepsilon_{i,N}(0)$.

Now similarly, define $\tilde G_N^{(1)}(\xi) = \frac{1}{\tilde n_T}\sum_{i=1}^N T_i \mathbbm{1}\{\varepsilon_{i,N}(1) \leq \xi\}$ and $\tilde G_N^{(2)}(\xi) = \frac{1}{n_T}\sum_{i=1}^N T_i \mathbbm{1}\{\varepsilon_{i,N}(1) \leq \xi\}$, and $\tilde F_N^{(1)}(\xi) = \frac{1}{\tilde n_{\bar T}}\sum_{i=1}^N \bar{T}_i \mathbbm{1}\{\varepsilon_{i,N}(0) \leq \xi\}$ and $\tilde F_N^{(2)}(\xi) = \frac{1}{n_{\bar T}}\sum_{i=1}^N \bar{T}_i \mathbbm{1}\{\varepsilon_{i,N}(0) \leq \xi\}$.

    \textit{(i) Weak convergence of the marginals  in probability}\\
The bounded-Lipschitz (or Fortet–Mourier) distance between two measures $\mu, \nu$
\begin{align}\nonumber
    d_{\text{BL}} (\mu,\nu) = \sup\bigg\{\int \phi d\mu  - \int \phi d\nu \, ;\,  \lVert \phi \rVert_\infty + \lVert \phi \rVert_{Lip} \leq 1 \bigg\}
\end{align}
metrizes weak convergence. Let BL be the space of all real-valued bounded Lipschitz functions (i.e., all $\phi: \R \rightarrow \R$ such that $\lVert \phi \rVert_\infty + \lVert \phi \rVert_{Lip} \leq 1$).
We have 
\begin{align}\nonumber
    d_{\text{BL}}(\tilde G_N^{(2)}, G) \leq d_{\text{BL}}(\tilde G_N^{(2)}, G_N) + d_{\text{BL}}(G_N, G) \rightarrow 0 \hspace{0.5cm} \text{in probability,}
\end{align}
where the second term converges to 0 because of Assumption \ref{as:convergJoint} and the first by Proposition \ref{Prop:GC} and Lemma \ref{lemma:BL}. Note that Assumption (b) implies that $\varepsilon_N(1)$ is $L_1$-bounded with respect to $G_N$ as required by Lemma \ref{lemma:BL}.

Additionally, we have  
\begin{align}
   d_{\text{BL}}(\tilde G_N^{(1)}, \tilde G_N^{(2)})\, &= \, \sup_{\phi \in \text{BL}} \E_{\tilde G_N^{(1)}}[ \phi(\varepsilon_{N})] - \E_{\tilde G_N^{(2)}}[ \phi(\varepsilon_{N})]\\\nonumber 
   &= \sup_{\phi\in \mathrm{BL}}\mathbbm{1}_{\mathcal D}\left(\frac{1}{\tilde n_T}-\frac{1}{n_T}\right)\sum_{i=1}^N T_i\phi(\varepsilon_{i,N}(1))+\mathbbm{1}_{\mathcal D^c}\phi(0)
 \\\nonumber
 &\leq \mathbbm{1}_{\mathcal D}\frac{|n_T-\tilde n_T|}{n_T}+\mathbbm{1}_{\mathcal D^c}
 = \frac{|n_T-\tilde n_T|}{n_T}.
    \end{align}
   Under Assumption \ref{as:RandomAssignment}.1, we then have $\frac{\lvert n_T - \tilde n_T\rvert}{n_T} = o_p(1)$. Under simple random sampling without replacement (Assumption \ref{as:RandomAssignment}.2), $\tilde n_T= n_T$ and the term vanishes. 
   We thus have $d_{\text{BL}}(\tilde G_N^{(1)}, G_N) \overset{p}{\rightarrow} 0$ by the triangle inequality, and  $\tilde G_N^{(1)} \overset{w}{\rightarrow} G$ in probability, where ``$\overset{w}{\rightarrow}$'' denotes weak convergence. Similarly, we have  $\tilde F_N^{(1)} \overset{w}{\rightarrow} F$ in probability.

   Now 
   \begin{align}
       d_{\text{BL}}(\hat G_N, \tilde G_N^{(1)})\, &= \, \sup_{\phi \in \text{BL}} \E_{\hat G_N}[ \phi(\hat \varepsilon_{N})] - \E_{\tilde G_N^{(1)}}[ \phi(\varepsilon_{N})] \\\nonumber
       \, &= \, 
      \sup_{\phi \in \text{BL}}  \frac{\mathbbm{1}_\mathcal{D}}{\tilde n_T} \sum_{i=1}^N T_i \phi(\hat \varepsilon_{i,N}(1)) -\frac{\mathbbm{1}_\mathcal{D}}{\tilde n_T} \sum_{i=1}^N T_i \phi(\varepsilon_{i,N}(1))
       \\\nonumber
       \, &\leq \, 
      \frac{\mathbbm{1}_\mathcal{D}}{\tilde n_T} \sum_{i=1}^N T_i \lvert \hat \varepsilon_{i,N} (1)-\varepsilon_{i,N}(1) \rvert = o_p(1)
   \end{align}
   by Assumption (a) and the fact that  $\mathbbm{1}_{\mathcal D}N/\tilde n_T=O_p(1)$ and
$\frac1N\sum_i|\hat\varepsilon_{i,N}(1)-\varepsilon_{i,N}(1)|^2=o_p(1)$.

Hence $d_{\text{BL}}(\hat G_N, G_N) \overset{p}{\rightarrow} 0$ by the triangle inequality, and  $\hat G_N \overset{w}{\rightarrow} G$ in probability. Similarly, we have  $\hat F_N \overset{w}{\rightarrow} F$ in probability.

\textit{(ii) Integration to the limit}\\
Assumption (b) allows us to integrate to the limit. Let
$\mathcal A_N=\{\tilde n_T\geq \pi_T N/2\}$. Under Assumption
\ref{as:RandomAssignment}, $\sum_N
P(\mathcal A_N^c)<\infty$ in the Bernoulli case; under simple random
sampling without replacement, $\mathcal A_N$ holds deterministically for all
sufficiently large $N$. Hence, by Borel--Cantelli, there exists an a.s. finite
random $N_0$ such that, for all $N\geq N_0$,
$\mathbbm{1}_{\mathcal D}/\tilde n_T\leq 2/(\pi_T N)$ a.s. Therefore,
\begin{align}\nonumber
   \sup_{N\geq N_0}
   \frac{\mathbbm{1}_\mathcal{D}}{\tilde n_T}
   \sum_{i=1}^N T_i \varepsilon_{i,N}(1)^2
   \mathbbm{1}\{\varepsilon_{i,N}(1)^2 \geq \beta\}
   \, &\leq \,
   \frac{2}{\pi_T}
   \sup_{N\geq N_0}
   \frac{1}{N}\sum_{i=1}^N
   \varepsilon_{i,N}(1)^2
   \mathbbm{1}\{\varepsilon_{i,N}(1)^2 \geq \beta\}
   \rightarrow 0 \quad \text{almost surely}
\end{align}
as $\beta\rightarrow\infty$. Hence Assumption (b) implies that the population-adjusted potential
outcomes $\varepsilon_{i,N}(1)$ are almost surely uniformly square-integrable
with respect to the sequence of random distributions $\{\tilde G_N^{(1)}\}$.
Similarly, $\varepsilon_{i,N}(0)$ are almost surely uniformly
square-integrable with respect to $\{\tilde F_N^{(1)}\}$.
Every subsequence has a further subsequence
$\{N_k\}$ along which
$\tilde G_{N_k}^{(1)}\overset{w}{\rightarrow}G$ and
$\tilde F_{N_k}^{(1)}\overset{w}{\rightarrow}F$ almost surely. By weak
convergence and the preceding uniform square-integrability,
$\E_{\tilde G_{N_k}^{(1)}}[\varepsilon_{N_k}(1)^2]
\rightarrow \E_G[\varepsilon(1)^2]$ almost surely. Similarly,
$\E_{\tilde F_{N_k}^{(1)}}[\varepsilon_{N_k}(0)^2]
\rightarrow \E_F[\varepsilon(0)^2]$ almost surely.

Let $r_{i,N}^{(1)}=f_{1,N}(X_i)-\hat f_{1,N}(X_i)$. Then we additionally have
\begin{align}
&\left|
\E_{\hat G_N}[\hat \varepsilon_N(1)^2]
-
\E_{\tilde G_N^{(1)}}[\varepsilon_N(1)^2]
\right| \\\nonumber
&\le
\frac{\mathbbm{1}_\mathcal D}{\tilde n_T}
\sum_{i=1}^N T_i r_{i,N}^{(1)2}
+
\frac{2\mathbbm{1}_\mathcal D}{\tilde n_T}
\sum_{i=1}^N T_i
|\varepsilon_{i,N}(1)r_{i,N}^{(1)}| \\\nonumber
&\le
\mathbbm{1}_\mathcal D\frac{N}{\tilde n_T}
\frac1N\sum_{i=1}^N r_{i,N}^{(1)2}
+
2\mathbbm{1}_\mathcal D\frac{N}{\tilde n_T}
\left(
\frac1N\sum_{i=1}^N \varepsilon_{i,N}(1)^2
\right)^{1/2}
\left(
\frac1N\sum_{i=1}^N r_{i,N}^{(1)2}
\right)^{1/2}
\rightarrow 0 \hspace{0.3cm} \text{in probability,}
\end{align}
where we used the Cauchy-Schwarz inequality. Here $\mathbbm{1}_{\mathcal D}N/\tilde n_T=O_p(1)$ by the preceding
denominator argument, $\frac1N\sum_i r_{i,N}^{(1)2}=o_p(1)$ by Assumption
(a), and $\frac1N\sum_i\varepsilon_{i,N}(1)^2=O(1)$ by Assumption (b).
Thus
$\E_{\hat G_N}[\hat\varepsilon_N(1)^2]
-
\E_{\tilde G_N^{(1)}}[\varepsilon_N(1)^2]=o_p(1)$.
Similarly,
$\E_{\hat F_N}[\hat\varepsilon_N(0)^2]
-
\E_{\tilde F_N^{(1)}}[\varepsilon_N(0)^2]=o_p(1)$.
We conclude 
\begin{align}
    \E_{\hat G_N}[\hat \varepsilon_N(1)^2] &\overset{P}{\rightarrow }\E_G[\varepsilon(1)^2] \hspace{1cm} \text{and} \\\nonumber
    \E_{\hat F_N}[\hat \varepsilon_N(0)^2] &\overset{P}{\rightarrow }\E_F[\varepsilon(0)^2].
\end{align}

Every subsequence has a further
subsequence $\{N_k\}$ along which both
$\hat G_{N_k}\overset{w}{\to}G$ and
$\E_{\hat G_{N_k}}[\hat\varepsilon_{N_k}(1)^2]\to
\E_G[\varepsilon(1)^2]$ hold almost surely. Along this further subsequence,
weak convergence plus convergence of second moments implies square-uniform integrability of
$\hat\varepsilon_{N_k}(1)$ with respect to $\{\hat G_{N_k}\}$. The same
argument applies to $\hat\varepsilon_{N_k}(0)$ with respect to
$\{\hat F_{N_k}\}$.

\textit{(iii) Convergence of the extremal joint distributions}\\
We define the extremal joint distributions $\Gamma^U(\xi_1,\xi_0) = \min\{G(\xi_1),F(\xi_0)\}$ and $\Gamma^L (\xi_1,\xi_0)= \max\{0,G(\xi_1)+F(\xi_0)-1\}$ and want to show that 
if the marginals converge weakly, so do 
$\hat \Gamma_N^U(\xi_1,\xi_0) = \min\{\hat G_N( \xi_1),\hat F_N(\xi_0)\}$ and $\hat \Gamma_N^L = \max\{0,\hat G_N(\xi_1)+\hat F_N(\xi_0)-1\}$. The Lévy metric
\begin{align}\nonumber
    d_L(F_1,F_2) = \inf\{\epsilon >0: F_1(x-\epsilon \boldsymbol{1}_d) - \epsilon \leq F_2(x) \leq F_1(x+\epsilon \boldsymbol{1}_d) + \epsilon \,   \forall x \in \R^d \}
\end{align}
metrizes weak convergence for multivariate CDFs. We have
\begin{align}
      d_L(\hat \Gamma^U_{N}, \Gamma^U) \, &=\,  \inf\{\epsilon >0 : \min\{\hat G_N(\xi_1 - \epsilon),\hat F_N(\xi_0 - \epsilon)\}- \epsilon \leq \min\{G(\xi_1),F(\xi_0)\} \\\nonumber \, &\hspace{1.3cm}\, \leq \min\{\hat G_N(\xi_1 + \epsilon),\hat F_N(\xi_0 + \epsilon)\}+ \epsilon  \,   \forall (\xi_1,\xi_0) \in \R^2 \} \\\nonumber
      \, & \, \leq 2 ( d_L(\hat G_N, G) + d_L(\hat F_N, F)) 
\end{align}
and 
\begin{align}
      d_L(\hat \Gamma^L_{N}, \Gamma^L) \, &=\,  \inf\{\epsilon >0 : \max\{0,\hat G_N(\xi_1 - \epsilon)+\hat F_N(\xi_0 - \epsilon)-1\}- \epsilon \\\nonumber
      \, &\hspace{1.3cm}\, \leq \max\{0,G(\xi_1)+F(\xi_0)-1\} 
      \\ \nonumber
      \, &\hspace{1.3cm}\, \leq \max\{0,\hat G_N(\xi_1 + \epsilon)+\hat F_N(\xi_0 + \epsilon)-1\}+ \epsilon  \,   \forall (\xi_1,\xi_0) \in \R^2 \} \\\nonumber
      \, & \, \leq 2 ( d_L(\hat G_N, G) + d_L(\hat F_N, F)). 
\end{align}
 Hence, we have
\begin{align}
    &d_L(\hat \Gamma^U_{N_k}, \Gamma^U) \rightarrow 0 \, \, \, \text{almost surely and } \\\nonumber
    &d_L(\hat \Gamma^L_{N_k}, \Gamma^L)  \rightarrow 0 \, \, \, \text{almost surely.}
\end{align}
We conclude $\hat \Gamma^U_{N} \overset{w}{\rightarrow} \Gamma^U$ and $ \hat \Gamma^L_{N} \overset{w}{\rightarrow} \Gamma^L$ in probability.

With $\hat \varepsilon_{N_k}(q)$ being uniformly square-integrable with respect  to $\{\hat G_{N_k}\}$ and $\{\hat F_{N_k}\}$ almost surely, we have  that $\hat \varepsilon_{N_k}(1) \hat \varepsilon_{N_k}(0)$ is uniformly integrable with respect to $\{\hat \Gamma^U_{N_k}\}$ and $\{\hat \Gamma^L_{N_k}\}$ almost surely.

Hence
\begin{align}
    &\E_{\hat \Gamma^U_N}[\hat \varepsilon_N(1) \hat \varepsilon_N(0)]\overset{P}{\rightarrow}  \E_{\Gamma^U}[ \varepsilon(1)\varepsilon(0)] = \sup_{\gamma \in \Pi(G,F)}\E_\gamma[\varepsilon(1)\varepsilon(0)] \, \, \, \text{and} \\\nonumber
    &\E_{\hat \Gamma^L_N}[\hat \varepsilon_N(1)\hat \varepsilon_N(0)] \overset{P}{\rightarrow}  \E_{\Gamma^L}[\varepsilon(1)\varepsilon(0)] = \inf_{\gamma \in \Pi(G,F)}\E_\gamma[\varepsilon(1)\varepsilon(0)].
\end{align}
The proposition now follows immediately.
\end{proof}

\section{Decorrelation Method for General Regression Adjustment}
\label{Sec:DecorRegAdj}
To illustrate how our sharp bounds apply to general regression adjustment beyond linear regression, we focus on the decorrelation estimator proposed by
\citet{Su_Mou_Ding_Wainwright_2023}. Their decorrelation method replaces the original treatment variables $Z_i$ and $1-Z_i$ with a decorrelation sequence $\{T_i,\bar T_i , R_i, \bar R_i\}$, see \citep{Su_Mou_Ding_Wainwright_2023} Lemma 1. The random variables $R$ and $\bar R$ determine on which units $\hat f_R \equiv \hat f_1$ and $\hat f_{\bar R} \equiv \hat f_0$ are estimated,
and $T_i$ and $\bar T_i$ replace $Z_i$ and $1-Z_i$ in the calculation of the estimator. Their oracle DC-estimator
\begin{align}\nonumber
    \hat \tau_{N}^{\text{dc.oracle}} := \frac{1}{n_T} \sum_{i=1}^N T_i(Y_i(1) - f_1(X_i)) - \frac{1}{n_{\bar T}} \sum_{i=1}^N \bar{T}_i(Y_i(0) - f_0(X_i)) + \frac{1}{N} \sum_{i=1}^N ( f_1(X_i) -  f_0(X_i)),
\end{align}
has the following variance under a BRE:
\begin{align}\nonumber
      \text{Var}(\hat \tau_N^{\text{dc.oracle}}) = \frac{1}{N} \bigg(\frac{1-\pi_T}{\pi_T} \frac{1}{N}\sum_{i=1}^N \varepsilon_i(1)^2 + \frac{1-\pi_{\bar{T}}}{\pi_{\bar{T}}}\frac{1}{N} \sum_{i=1}^N \varepsilon_i(0)^2 + \frac{2}{N} \sum_{i=1}^N\varepsilon_i(1) \varepsilon_i(0)\bigg).
\end{align}
Their decorrelation estimator is
\begin{align}\nonumber
    \hat \tau_{N}^{\text{dc}} := \frac{1}{n_T} \sum_{i=1}^N T_i(Y_i(1) - \hat f_R(X_i)) - \frac{1}{n_{\bar T}} \sum_{i=1}^N \bar{T}_i(Y_i(0) - \hat f_{\bar{R}}(X_i)) + \frac{1}{N} \sum_{i=1}^N (\hat f_R(X_i) - \hat f_{\bar{R}}(X_i))
\end{align}
and they propose the following identified, conservative variance estimator based on Neyman's upper bound 
\begin{align}\nonumber
\hat V_N^{dc} = \frac{1}{N\pi_T^2} \sum_{i=1}^N T_i (Y_i(1)-\hat{f}_R(X_i))^2 + \frac{1}{N\pi_{\bar T}^2} \sum_{i=1}^N \bar{T}_i (Y_i(0)-\hat{f}_{\bar R}(X_i))^2.
\end{align}
This estimator is of the form of the general regression estimator covered in Appendix \ref{App:Theory}.
Under the conditions of their proposition guaranteeing asymptotic normality (\citep{Su_Mou_Ding_Wainwright_2023} Proposition 1), our Proposition \ref{Prop:ConsistencyOfSharpBounds_star} provides sharp variance bounds, as given in Appendix \ref{App:Theory}.
\begin{corollary} 
Let Assumptions \ref{as:convergJoint} and \ref{as:RandomAssignment}.1 hold. Further, assume
\begin{enumerate} [label=(\alph*)]
    \item that the quadruple $\{T,\bar{T},R,\bar{R}\}$ is a decorrelation sequence as in Lemma 1 in \citep{Su_Mou_Ding_Wainwright_2023}.
    \item $\liminf_{N \rightarrow \infty}. \text{Var}(\hat \tau_N^{dc.oracle}) > 0 $
    \item The fourth moment of the population-adjusted outcomes is bounded: $\frac{1}{N}\sum_{i=1}^N \varepsilon_{i,N}^4(q) \leq C$ 
    for $q \in \{0,1\}$.
    \item  The treatment probability remains uniformly bounded away from 0, i.e., $\pi_T, \pi_{\bar T}\in [\alpha,1-\alpha]$ for some $\alpha \in (0,1/2)$ independent of N. 
    \item We have $o_p(1)$ consistency of the outcome models, i.e.,
    $\big( \frac{1}{N} \sum_{i=1}^N (f_1(X_i) - \hat f_R(X_i))^2 \big)^{1/2}= o_p(1)$ and $\big( \frac{1}{N} \sum_{i=1}^N (f_0(X_i) - \hat f_{\bar R}(X_i))^2 \big)^{1/2}= o_p(1)$.
\end{enumerate}
Then 
\begin{align}\nonumber
    \frac{(\hat \tau^{dc}_N - \tau )}{\sqrt{\kappa \widehat{\emph{SVB}}_N}} \overset{d}{\rightarrow} \mathcal{N}(0,1)
\end{align}
with $\kappa \leq 1$.
\end{corollary}
\begin{proof}
    The asymptotic normality is proven in \citep{Su_Mou_Ding_Wainwright_2023}. 
    Condition (a) implies that our Assumption \ref{as:RandomAssignment} about the random indicators holds. 
    Condition (c) immediately implies that $\varepsilon_N(1)$ and $\varepsilon_N(0)$ are uniformly square-integrable. Thus, the conditions of Proposition \ref{Prop:ConsistencyOfSharpBounds_star} are met, and the corollary follows.
\end{proof}
\begin{remark}
    Using these asymptotic results, researchers can use the 
    asymptotically conservative Wald-type confidence interval
      $ \hat \tau^{dc}_N \pm z_{1-\alpha/2} \sqrt{\widehat{\emph{SVB}}_N}$ using the sharp variance bound estimator derived in this paper.
  \end{remark} 
\section{Practical Considerations regarding Benefits of Sharp Bounds}
\label{App:BenSharpBounds}
\begin{figure}
\centering\includegraphics[width=0.5\textwidth]{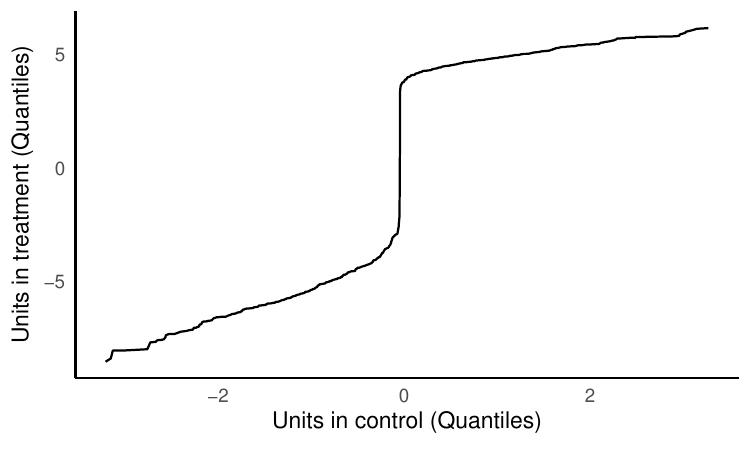}
        \caption{ Q-Q plot of the sample-adjusted potential outcomes in treatment and control for our simulation in Section \ref{Sec:SimRes} for $\theta = 0.5$.}
        \label{fig:QQplots}
\end{figure}
In this section, we compare the sharp variance estimator for regression adjustment with the conventional and Cauchy-Schwarz variance estimators. For ease of exposition, we focus on linear regression adjustment for which the sample-adjusted potential outcomes are mean-centered. This allows us to switch between second moments and covariances.

Comparing our sharp variance estimator (Equation \ref{Eq:estmVarBoundsReg}) with the Cauchy-Schwarz variance estimator (Equation \ref{Eq:CSVarEst}), we see that they only differ in how they treat the unidentified cross-term. 
For the sharp variance estimator, the term is estimated based on the cross-moment of the sample-adjusted potential outcomes under the extremal joint distribution (i.e., assuming the sample-adjusted outcomes to be comonotonic). This cross-moment will match the estimate in the Cauchy-Schwarz estimator if and only if the Cauchy-Schwarz inequality yields an equality, that is, when the sample-adjusted potential outcomes are linearly dependent. Deviations from linear dependence can be detected by inspection of Q-Q plots. Figure \ref{fig:QQplots} shows the Q-Q plot for the sample-adjusted potential outcomes from the simulation performed in Section \ref{Sec:SimRes} for $\theta =0.5$, and Figure \ref{fig:QQplots_experiments} presents the Q-Q plots for the experiments in \citep{Murrar_18} and \citep{Lee_22}, which we discuss in Section \ref{Sec:Appl}. Both plots clearly show deviations from the diagonal, apropos of the efficiency gains we found numerically in Figure \ref{fig:simulationResults} and Tables \ref{tab:results_applc_murrar} and \ref{tab:results_applc_lee}.

In the case of linear regression, the sample-adjusted potential outcomes have zero mean. If we further assume a balanced design, the cross-moment under the extremal joint (Equation \ref{Eq.CalcBounds})
\begin{align}\nonumber
     \E_{ \hat \Gamma^U_N}(\hat\varepsilon_{N}(1)\hat\varepsilon_{N}(0)) = \sum_{i=1}^P (p_i - p_{i-1})\hat \varepsilon_{(\lceil \tilde n_{T} p_i \rceil)1}^{obs}\hat \varepsilon_{(\lceil \tilde n_{\bar T} p_i \rceil)0}^{obs} = \text{Cov}_{\hat \Gamma^U_N}(\hat \varepsilon^{obs}_1,\hat \varepsilon^{obs}_0)
\end{align}
reduces to the covariance between the sorted sample-adjusted potential outcomes. Defining $\rho$ to be the correlation between the sorted sample-adjusted potential outcomes in treatment and control, we have
\begin{align}\nonumber
    \frac{\E_{ \hat \Gamma^U_N}(\hat\varepsilon_{N}(1)\hat\varepsilon_{N}(0))}{ (\frac{1}{n_1}\sum_{i=1}^{n_1}\hat \varepsilon_{i,N}(1)^2)^{1/2}(\frac{1}{n_0} \sum_{i=1}^{n_0}\hat\varepsilon_{i,N}(0)^2)^{1/2}} = \rho.
\end{align}
The ratio of the cross-term estimate of the sharp variance estimator and the Cauchy-Schwarz estimator is given by the correlation $\rho$ of the sorted sample-adjusted potential outcomes.
In balanced designs, the correlation between sorted sample-adjusted outcomes provides an additional diagnostic tool beyond Q-Q plots to gauge the benefits of the sharp variance estimator compared to the Cauchy-Schwarz variance estimator.

The conventional variance estimator (Equation \ref{eq:convVarEst}) is looser than the Cauchy-Schwarz variance estimator when the AM-GM inequality is loose. That is the case when the variances in treatment and control differ. 
It is straightforward to see, however, that very uneven variances in treatment and control lead the variance of the general regression estimator (Equation \ref{Eq:VarOracle}) to be dominated by the identified part. Let $r = \frac{\sigma_1}{\sigma_0}$ be the ratio of the standard deviation of the sample-adjusted potential outcomes in treatment and control. We then have
\begin{align}\nonumber
    \frac{N\widehat{\text{SVB}}_N}{\frac{n_0}{n_1} \E_{\hat G_N}\big[\hat\varepsilon_N(1)^2\big] + \frac{n_1}{n_0}\E_{\hat F_N}\big[\hat\varepsilon_N(0)^2\big]} = 
    \frac{\frac{n_0}{n_1} r^2+ \frac{n_1}{n_0} + 2 r \rho}{\frac{n_0}{n_1} r^2+ \frac{n_1}{n_0}} \rightarrow 1
\end{align}
as $r \rightarrow \infty$. Hence, with significantly different variances in treatment and control, the differences between the sharp and  Cauchy-Schwarz variance estimators become negligible. 

Similarly, the differences among these estimators become less relevant for extremely unbalanced designs. We have
\begin{align}\nonumber
    \frac{N\widehat{\text{SVB}}_N}{\frac{n_0}{n_1} \E_{\hat G_N}\big[\hat\varepsilon_N(1)^2\big] + \frac{n_1}{n_0}\E_{\hat F_N}\big[\hat\varepsilon_N(0)^2\big]} \rightarrow 1
\end{align}
if either $\frac{n_{0}}{n_1} \rightarrow \infty$ or $\frac{n_1}{n_0} \rightarrow \infty$.  Biomedical randomized controlled trials tend to use balanced designs, highlighting one key domain that may benefit from sharp variance bounds.

\section{Sharp Lower Bounds}
\label{App:LowerBounds}
By the same logic by which we derived sharp upper bounds on the variance of treatment effect estimators, we may also derive sharp lower variance bounds. Sharp lower bounds have been derived for the difference-in-means estimator by \citet{Aronow_Green_Lee_2014}.

Using the Fréchet-Hoeffding copula bounds from Theorem \ref{ref:FH}, we immediately have the following lemma characterizing a lower bound on the expectation of a product of two random variables:
\begin{lemma}
\label{lemma:Hoeffding_lower}
For two random variables  $X_1$ and $X_2$ with marginal CDFs $G$ and $F$, we have
    \[
    \min_{\gamma \in \Pi(G,F)} \E_{(X_1,X_2) \sim \gamma}[X_1X_2] = \E_{(X_1,X_2) \sim \Gamma^L}[X_1X_2]
    \]
with $\Gamma^L(x_1,x_2) = \max\{-1 +G(x_1)+F(x_2),0\}$.
   The lower bound is attained when $X_1$ and $X_2$ are counter-monotonic.
\end{lemma}
The sharpness of the preceding result is a consequence of the lower Fréchet-Hoeffding bound being a copula in two dimensions \citep{Joe_2014}.

Based on this result, we can then derive a sharp lower bound on the variance (LSVB) and a related estimator. Under a CRE, we have
\begin{align*}
      \text{LSVB}_N & = \frac{1}{N-1} \bigg( \frac{N-n_T}{n_T} \text{Var}_{G_N}\big(\varepsilon(1)\big) + \frac{N-n_{\bar T}}{n_{\bar T}} \text{Var}_{F_N}\big(\varepsilon(0)\big)  
       +  2 \text{Cov}_{\Gamma^L_N}\big(\varepsilon(1),\varepsilon(0)\big)\bigg) \enspace \text{and} \\
        \widehat{\text{LSVB}}_N & = \frac{1}{N-1} \bigg( \frac{N-n_T}{n_T} \text{Var}_{\hat G_N}\big(\hat \varepsilon(1)\big) + \frac{N-n_{\bar T}}{n_{\bar T}} \text{Var}_{\hat F_N}\big(\hat \varepsilon(0)\big)  
       +  2 \text{Cov}_{\hat \Gamma^L_N}\big(\hat \varepsilon(1),\hat \varepsilon(0)\big)\bigg).
\end{align*}
The consistency of the estimator for the sharp lower bound on the variance $\widehat{\text{LSVB}}_N$ follows analogously to the proof of the consistency of $\widehat{\text{SVB}}_N$, see Appendix \ref{Appendix:proof_prop1}.

The key ingredient $\E_{\hat \Gamma^L_N}[\hat \varepsilon_N(1)\hat \varepsilon_N(0)] \overset{P}{\rightarrow}  \E_{\Gamma^L}[\varepsilon(1)\varepsilon(0)] = \inf_{\gamma \in \Pi(G,F)}\E_\gamma[\varepsilon(1)\varepsilon(0)]$ 
is proved in Appendix \ref{Appendix:proof_prop1}.